%% file: main-arxiv.tex
\documentclass[a4paper,UKenglish]{llncs}

\usepackage[utf8]{inputenc}
\usepackage[english]{babel}

\usepackage{amssymb,amsmath,amsfonts}
\usepackage[usenames,dvipsnames]{xcolor}

\newcommand{\citeapp}[1]{\cite{myproofs}}

\usepackage{enumerate}
\usepackage{wrapfig}
\usepackage{xspace}
\usepackage{stmaryrd}
\usepackage{textcomp} 

\usepackage{tikz,pgffor}
\usetikzlibrary{arrows}
\usetikzlibrary{shapes}
\usetikzlibrary{automata}
\usetikzlibrary{decorations.pathmorphing} 
\usetikzlibrary{decorations.pathreplacing} 
\usetikzlibrary{calc} 

\newcommand{\Nset}{\mathbb{N}}
\newcommand{\Nseto}{\Nset_0}

\newcommand{\Rset}{\mathbb{R}}

\newcommand{\Rsetpo}{\mathbb{R}_{\ge 0}}

\newcommand{\powerset}{\mathcal P}
\newcommand{\nxt}{\mathrm{next}}
\newcommand{\cannot}{\bullet}

\newcommand{\dist}{\mathcal{D}}
\newcommand{\measures}{\mathcal{D}}
\renewcommand{\vec}[1]{\mathbf{#1}}
\newcommand{\de}[1]{\mathit{d#1}}  %

\newcommand{\ex}{\mathrm{Exp}}

\newcommand{\B}{E}
\newcommand{\choices}{\mathcal W}

\newcommand{\corners}{\mathcal E}
\usepackage[noend]{algorithm2e}

\newcommand{\bisim}{\sim}
\newcommand{\hennessy}{\mathit{Hen}}

\newcommand{\sigmafield}{\Sigma}

\newcommand{\states}{\mathit{S}}
\newcommand{\sta}{s}
\newcommand{\Sta}{X}
\newcommand{\statesfield}{\sigmafield(\states)}

\newcommand{\labels}{\mathit{L}}
\newcommand{\lab}{a}
\newcommand{\Lab}{A}
\newcommand{\labelsfield}{\sigmafield(\labels)}

\newcommand{\resolver}{\rho}

\newcommand{\fail}{\square}

\newcommand{\dis}{\mu}
\newcommand{\distwo}{\nu}
\newcommand{\Dis}{M}

\newcommand{\mea}{\mu}
\newcommand{\meatwo}{\nu}
\newcommand{\Mea}{M}

\newcommand{\statemeasuresfield}{\sigmafield(\measures(\states))}

\newcommand{\functorfin}{\heartsuit}
\newcommand{\functorinf}{\spadesuit}

\newcommand{\act}{L}

\newcommand{\tran}[1]{{}\mathchoice%
    {\xrightarrow{#1}}
    {\smash{\xrightarrow{#1}}}
    {\xrightarrow{#1}}
    {\xrightarrow{#1}}
{}}

\newcommand{\arrowsub}[1]{\hspace{-2pt}_{#1}}

\newcommand{\suc}{\mathrm{succ}}

\newcommand{\prob}{\mathbf{P}}

\newcommand{\probm}{\mathcal{P}}

\newcommand{\kernel}{\tau}

\newcommand{\PA}{D}
\newcommand{\patrans}{\longrightarrow}

\newcommand{\patransa}[1]{{}\mathchoice%
    {\stackrel{#1}{\longrightarrow}}
    {\mathop {\smash\longrightarrow}\limits^{\vrule width 0pt height 0pt depth 4pt\smash{#1}}}
    {\stackrel{#1}{\longrightarrow}}
    {\stackrel{#1}{\longrightarrow}}
{}}

\newcommand{\pts}{\mathbf{P}}

\newcommand{\SA}{\mathcal{S}}
\newcommand{\locations}{\mathcal{Q}}
\newcommand{\clocks}{\mathcal{C}}
\newcommand{\actions}{\mathcal{A}}
\newcommand{\edges}{\rightarrow}
\newcommand{\clocksetting}{\kappa}
\newcommand{\clockdist}{F}

\newcommand{\loc}{q}
\newcommand{\clo}{c}
\newcommand{\Clo}{X}
\newcommand{\val}{\xi}

\newcommand{\loctop}{q}
\newcommand{\locleft}{u}
\newcommand{\locright}{v}
\newcommand{\locbottom}{r}

\newcommand{\chain}{\bar{\PA}}
\newcommand{\chainstates}{\bar{S}}
\newcommand{\SAstates}{S}

\newcommand{\abssemantics}{\hat{\pts}}
\newcommand{\absstates}{\hat{\states}}

\newcommand{\todo}[1]{\todoc{#1}{red}}
\newcommand{\todoc}[2]{%
  \begin{tikzpicture}[remember picture]%
      \node [coordinate] (inText) {};%
  \end{tikzpicture}%
  \marginpar{%
      \begin{tikzpicture}[remember picture]%
          \draw node[draw=#2, color=#2, text width = 4.2cm, inner sep=2pt] (inNote)%
                   {{\footnotesize #1}};%
      \end{tikzpicture}%
  }%
  \begin{tikzpicture}[remember picture, overlay]%
      \draw[draw = #2, thick]%
          ([yshift=-0.1cm] inText)%
              -| ([xshift=-0.1cm] inNote.west)%
              -| (inNote.west);%
  \end{tikzpicture}%
}%
\renewcommand{\todoc}[2]{}
\newcommand{\holger}[1]{\todoc{Holger: #1}{WildStrawberry}}
\newcommand{\jena}[1]{\todoc{Jeňa: #1}{OliveGreen}}

\newcommand{\out}[1]{} 

\newcommand{\tableau}[1]{\textbf{#1}}

\newcommand{\theoremlike}[2]{\par\medskip\penalty-250\refstepcounter{theorem}{{\bfseries\noindent#2
\ref{#1}.}}}

\newcommand{\thmhelperpre}[2]{\theoremlike{#1}{#2}}
\newcommand{\thmhelperpost}{\par\medskip}

\newenvironment{reftheorem}[1]{\thmhelperpre{#1}{Theorem}}{\thmhelperpost}
\newenvironment{reflemma}[1]{\thmhelperpre{#1}{Lemma}}{\thmhelperpost}

\newenvironment{refproposition}[1]{\thmhelperpre{#1}{Proposition}}{
\thmhelperpost }

\newif\iffull
\fulltrue
\newcommand{\refappendix}{the appendix}
\newcommand{\QED}{\qed}

\newcommand{\myspace}{\vspace*{-1em}}

\usepackage{myspacehack}

\title{Probabilistic Bisimulation: \\ 
Naturally on Distributions}

\author{ 
Holger Hermanns\inst{1} \and Jan Kr\v{c}\'{a}l\inst{1} \and Jan K\v{r}et\'{i}nsk\'{y}\inst{2}
}

\institute{
Saarland University -- Computer Science, Saarbr\"ucken, Germany
    \texttt{\{hermanns,krcal\}@cs.uni-saarland.de}
\and 
IST Austria
	\texttt{jan.kretinsky@ist.ac.at}
}

\begin{document}

\pagestyle{plain}

\maketitle

\setcounter{footnote}{0}
\renewcommand{\thefootnote}{(\arabic{footnote})}

\begin{abstract}
%
  In contrast to the usual understanding of probabilistic systems as
  stochastic processes, recently these systems have also been regarded
  as transformers of probabilities. In this paper, we give a
  \emph{natural} definition of strong bisimulation for probabilistic
  systems corresponding to this view that treats probability
  \emph{distributions} as first-class citizens. Our definition applies
  in the same way to discrete systems as well as to systems with
  uncountable state and action spaces.  Several examples demonstrate
  that our definition refines the understanding of behavioural
  equivalences of probabilistic systems. In particular, it solves a
  long-standing open problem concerning the representation of
  memoryless continuous time by memory-full continuous time. Finally,
  we give algorithms for computing this bisimulation not only for
  finite but also for classes of uncountably infinite systems.
\end{abstract}

\input{intro-alternative}

\input{bisim-new}

\input{applications}

\input{algorithms}
\input{rw}

\section{Conclusion}

We have introduced a general and natural notion of a distribution-based probabilistic bisimulation, shown its applications in different settings and given algorithms to compute it for finite and some classes of infinite systems. As to future work, the precise complexity of the finite case is certainly of interest. Further, the tableaux decision method opens the arena for investigating wider classes of continuous-time systems where the new bisimulation is decidable.

\bibliographystyle{myabbrv}
\bibliography{main}

\appendix
\newpage
\input{app}

\end{document}

%% file: intro-alternative.tex
\section{Introduction} 
\label{sec:intro}

Continuous time concurrency phenomena can be addressed in two principal
manners: On the one hand, \emph{timed automata} (TA) extend
interleaving concurrency with real-valued clocks~\cite{alur-dill}.
On the other hand, time can be represented by memoryless stochastic
time, as in \emph{continuous time Markov chains} (CTMC) and extensions,
where time is represented in the form of exponentially distributed
random delays~\cite{pepa,HHM,ymca,DBLP:conf/lics/EisentrautHZ10}. TA and CTMC variations
have both been applied to very many intriguing cases, and are
supported by powerful real-time, respectively stochastic time model
checkers~\cite{uppaal,prism} with growing user bases.  The models are
incomparable in expressiveness, but if one extends timed automata with
the possibility to sample from exponential distributions~\cite{tutti,DBLP:journals/iandc/DArgenioK05,strulo}, there
appears to be a natural bridge from CTMC to TA. This kind of
stochastic semantics of timed automata has recently gained
considerable popularity by the statistical model checking approach to
TA analysis~\cite{statistical1,statistical2}.

Still there is a disturbing difference, and this difference is the
original motivation~\cite{pedro-christel} of the work presented in
this paper. The obvious translation of an exponentially distributed
delay into a clock expiration sampled from the very same exponential
probability distribution fails in the presence of concurrency. This is
because the translation is not fully compatible with the natural
interleaving concurrency semantics for TA respectively CTMC. This is
illustrated by the following example, which in the middle displays two
small CTMC, which are supposed to run independently and
concurrently.\jena{mention that 1 and 2 are rates?}\holger{Lengthy to
  explain it here to someone who does not know it already, and those
  who know do not need the explanation.}

\begin{center}
\myspace
\begin{tikzpicture}[outer sep=0.1em,->,
state/.style={draw,circle,minimum size=1.7em, inner sep=0.1em},
sstate/.style={draw,circle,minimum size=0.7em, inner sep=0.1em},
label/.style={font=\small}]

\begin{scope}[xshift=-2.6cm]
\node (s) at (0,2) [state] {$\loctop$};
\node (t) at (-1,1) [state] {$\locleft$};
\node (u) at (1,1) [state] {$\locright$};
\node (v) at (0,0) [state] {$\locbottom$};
\node (ss) at (1,2.1) {$\substack{x:=\ex(1),\\y:=\ex(2)}$};

\path[->]
(s) edge node[label,right]{$a$} node[label,left]{$x=0$} (t)
(s) edge node[label,left]{$b$} node[label,right]{$y=0$} (u)
(t) edge node[label,right]{$b$} node[label,left]{$y=0$} (v)
(u) edge node[label,left]{$a$} node[label,right]{$x=0$} (v)
;
\end{scope}

\begin{scope}[xshift=4.7cm]
\node (h) at (-4,1.5) [sstate] {$$};
\node (i) at (-4,0.5) [sstate] {$$};

\node (j) at (-3,1.5) [sstate] {$$};
\node (k) at (-3,0.5) [sstate] {$$};

\path[->]
(h) edge node[label,right]{1}  (i)
(j) edge node[label,right]{2}  (k)
;
\end{scope}

\begin{scope}[xshift=5.6cm]
\node (s) at (0,2) [state,inner sep=0em] {$\loctop'$};
\node (t) at (-1,1) [state,inner sep=0em] {$\locleft'$};
\node (u) at (1,1) [state,inner sep=0em] {$\locright'$};
\node (v) at (0,0) [state,inner sep=0em] {$\locbottom'$};
\node (ss) at (1,2.1) {$\substack{x:=\ex(1),\\y:=\ex(2)}$};
\node (tt) at (-2,1) {$\substack{\\y:=\ex(2)}$};
\node (uu) at (2,1) {$\substack{\\x:=\ex(1)}$};

\path[->]
(s) edge node[label,right]{$a$} node[label,left]{$x=0$} (t)
(s) edge node[label,left]{$b$} node[label,right]{$y=0$} (u)
(t) edge node[label,right]{$b$} node[label,left]{$y=0$} (v)
(u) edge node[label,left]{$a$} node[label,right]{$x=0$} (v)
;
\end{scope}
\end{tikzpicture}
\myspace
\end{center}
On the left and right we see two stochastic automata (a variation of
timed automata formally defined in Section~\ref{sec:applications!}). They have
clocks $x$ and $y$ which are initialized by sampling from exponential
distributions, and then each run down to $0$. The first one reaching
$0$ triggers a transition and the other clock keeps on running unless
resampled, which happens on the right, but not on the left. The left
model is obtained by first translating the respective CTMC, and then
applying the natural TA interleaving semantics, while the right model
is obtained by first applying the equally natural CTMC interleaving
semantics prior to translation.

The two models have subtly different semantics in terms of their
underlying dense probabilistic timed transition systems. This can
superficially be linked to the memoryless property of exponential
distributions, yet there is no formal basis for proving equivalence.
This paper closes this gap, which has been open for at least 15 years,
by introducing a natural \emph{continuous-space distribution-based}
bisimulation. This result is embedded in several further intriguing
application contexts and algorithmic achievements for this novel
bisimulation\holger{rephrased}.

The theory of bisimulations is a well-established and elegant
framework to describe equivalence between processes based on their
behaviour. In the standard semantics of probabilistic
systems~\cite{DBLP:conf/popl/LarsenS89,SegalaL94}, when a
probabilistic step from a state to a distribution is taken, the random
choice is resolved and we instead continue from one of the successor states.
Recently, there has been considerable interest in instead regarding
probabilistic systems as deterministic transformers of probability
\emph{distributions}~\cite{DBLP:conf/qest/KorthikantiVAK10,DBLP:conf/lics/AgrawalAGT12,DBLP:journals/corr/0001MS13a}, where the choice is not resolved and we continue from the
distribution over successors.  Thus, instead of the current state the
transition changes the current distribution over the states. Although
the distribution semantics is very natural in many contexts
\cite{DBLP:dblp_journals/fac/Hennessy12}, it has been only partially
reflected in the study of bisimulations
\cite{DBLP:dblp_journals/fac/Hennessy12,DBLP:journals/ijfcs/DoyenHR08,DBLP:journals/corr/FengZ13,DBLP:conf/lics/EisentrautHZ10}.

Our definition arises as an unusual, but very simple instantiation of
the standard coalgebraic framework for bisimulations
\cite{Sangiorgi:2011:ATB:2103601}. (No knowledge of coalgebra is
required from the reader though.) Despite its simplicity, the
resulting notion is surprisingly fruitful, not only because it indeed
solves the longstanding correspondence problem between CTMC and TA with
stochastic semantics.

Firstly, it is more adequate than other equivalences when applied to
systems with distribution semantics, including large-population models
where different parts of the population act differently
\cite{may1974biological}. Indeed, as argued in
\cite{DBLP:journals/fac/GeorgievskaA12}, some equivalent states are
not identified in the standard probabilistic bisimulations and too
many are identified in the recent distribution based bisimulations
\cite{DBLP:journals/ijfcs/DoyenHR08,DBLP:journals/corr/FengZ13}. Our
approach allows for a bisimulation identifying precisely the desired
states~\cite{DBLP:journals/fac/GeorgievskaA12}.

Secondly, 
our bisimulation over distributions induces an equivalence \emph{on
  states}, and  this relation equates behaviourally indistinguishable
states which in many settings are unnecessarily distinguished by standard bisimulations. We shall discuss this phenomenon in the context of several
applications. \holger{What we are saying here is that the new relation is
  coarser, and that this is good, right?} \holger{And I do not get the
difference between 1 and 2.}
Nevertheless, the key idea to work with distributions instead of
single states also bears disadvantages. The main difficulty is that
even for finite systems the space of distributions is uncountable,
thus bisimulation is difficult to compute. However, we show that it
admits a concise representation using methods of linear algebra and we
provide an  algorithm for computing it.  Further, in order to cover e.g.\
continuous-time systems, we need to handle both uncountably many
states (that store the sampled time) and labels (real time
durations). Fortunately, there is an elegant way to do so using the
standard coalgebra framework\holger{This is pretty much hidden now. Deephasize}. Moreover, it can easily be further
generalized, e.g.\ adding rewards to the generic definition is a
trivial task.

\noindent \textbf{Our contribution} is the following:
\begin{itemize}
 \item We give a natural definition of bisimulation from the distribution perspective for systems with generally uncountable spaces of states and labels.
 \item We argue by means of several applications that the definition
   can be considered more useful than the classical notions of probabilistic bisimulation.
 \item We provide an algorithm to compute this distributional
   bisimulation on finite non-deterministic probabilistic systems, and
   present a decision algorithm for uncountable continuous-time
   systems induced by the stochastic automata mentioned above.
\end{itemize}
\iffull Full proofs can be found in the appendix. \else A full version of this paper is available~\refappendix. \fi 


%% file: bisim-new.tex
\section{Probabilistic bisimulation on distributions}
\label{sec:bisim}

A (potentially uncountable) set $S$ is a \emph{measurable space} if it is equipped with a $\sigma$-algebra, which we denote by $\sigmafield(X)$. The elements of $\sigmafield(X)$ are called \emph{measurable sets}.
For a measurable space $S$, let $\measures(S)$ denote the set of \emph{probability measures} (or \emph{probability distributions}) over $S$. The following definition is similar to the treatment of~\cite{phd-cont}.
\todo{``The following definition is similar to the treatment of~\cite{phd-cont}.'' to related work?}
\jena{I'd prefer to keep it here. We did not come up with this Def.}

\begin{definition}
 A \emph{non-deterministic labelled Markov process} (NLMP) is a tuple $\pts = (\states, \labels, \{\kernel_\lab \mid \lab\in\labels\})$ where $\states$ is a measurable space of \emph{states}, $\labels$ is a measurable space of \emph{labels}, and 
 $\kernel_\lab : \states \to \statemeasuresfield$ assigns to each state $\sta$ a measurable set of probability measures $\kernel_\lab(\sta)$ available in $\sta$ under $\lab$.%
 \footnote{We further require that for each $\sta \in \states$ we have $\{(\lab,\mea) | \mea \in \kernel_\lab(\sta)\} \in \labelsfield \otimes \statemeasuresfield$ and for each $\Lab \in \labelsfield$ and $Y \in \statemeasuresfield$ we have $\{\sta\in\states \mid \exists \lab \in \Lab. \kernel_\lab(\sta) \cap Y \neq \emptyset \}\in \statesfield$. Here $\statemeasuresfield$ is the Giry $\sigma$-algebra~\cite{giry1982categorical} over $\measures(X)$.}
\end{definition}

When in a state $\sta\in\states$, NLMP reads a label $\lab \in \labels$ and \emph{non-deterministically} chooses a successor distribution $\dis\in\measures(\states)$ that is in the set of convex combinations\footnote{
A distribution $\mea \in\measures(\states)$ is a \emph{convex combination} of a set $\Mea \in \statemeasuresfield$ of distributions if there is a measure $\meatwo$ on $\measures(\states)$ such that $\meatwo(\Mea) = 1$ and $\mea = \int_{\mea'\in\measures(\states)} \mea' \meatwo(\de{\mea'})$.
} 
over $\kernel_\lab(\sta)$, denoted by $s \patransa{\lab} \dis$. If there is no such distribution, the process halts. Otherwise, it moves into a successor state according to $\dis$. 
Considering convex combinations is necessary as it gives more power than pure resolution of non-determinism \cite{Segala:1996:MVR:239648}.

\begin{example}
If all sets are finite, we obtain \emph{probabilistic automata (PA)} defined \cite{Segala:1996:MVR:239648} as a triple $(\states, \labels, \patrans)$ where $\mathord{\patrans} \subseteq \states \times \labels \times \dist(\states)$ is a probabilistic transition relation with $(\sta,\lab,\dis)\in\mathord{\patrans}$ if $\dis\in\kernel_\lab(\sta)$.
\end{example}

\begin{example}\label{ex:cont}
In the continuous setting, consider a random number generator that also remembers the previous number. We set $\labels=[0,1]$, $\states = [0,1]\times[0,1]$ and $\kernel_{x}(\langle \textit{new},\textit{last}\rangle) = \{\mea_{x}\}$ for $x=\mathit{new}$ and $\emptyset$ otherwise, where $\mea_{x}$ is the uniform distribution on $[0,1]\times \{x\}$. If we start with a uniform distribution over $S$, the measure of successors under any $x\in\labels$ is $0$. Thus in order to get any information of the system we have to consider successors under sets of labels, e.g.\ intervals.
\end{example}


For a measurable set $\Lab\subseteq\labels$ of labels, we write $\sta \patransa{\Lab} \mea$ if $\sta \patransa{\lab} \mea$ for some $\lab\in\Lab$, and denote by $\states_\Lab := \{\sta \mid \exists \mea: \sta \patransa{\Lab} \mea\}$ the set of states having some outgoing label from $\Lab$.
Further, we can lift this to probability distributions by setting $\mea \patransa{\Lab} \meatwo$ if 
$\meatwo = \frac{1}{\mea(\states_\Lab)}\int_{\sta \in\states_\Lab} \meatwo_\sta\ \mea(d\,\sta)$
for some measurable function assigning to each state $\sta\in\states_\Lab$ a measure $\meatwo_\sta$ such that $\sta \patransa{\Lab} \meatwo_\sta$.
Intuitively, in $\mea$ we restrict to states that do not halt under $A$ and consider all possible combinations of their transitions; we scale up by $\frac{1}{\mea(\states_\Lab)}$ to obtain a distribution again.
\jena{maybe citing hyper-transitions from the PhD thesis Pedro has sent}

\begin{example}
In the previous example, let $\upsilon$ be the uniform distribution. Due to the independence of the random generator on previous values, we get $\upsilon\patransa{[0,1]}\upsilon$. Similarly, $\upsilon\xrightarrow{[0.1,0.2]}\upsilon_{[0.1,0.2]}$ where $\upsilon_{[0.1,0.2]}$ is uniform on $[0,1]$ in the first component and uniform on $[0.1,0.2]$ in the second component, with no correlation.
\end{example}

Using this notation, a non-deterministic and probabilistic system such
as NLMP can be regarded as a non-probabilistic, thus solely
non-deterministic\holger{how about this?}, labelled transition system
over the uncountable space of probability distributions.  The natural
bisimulation from this distribution perspective is as follows.

\begin{definition}\label{def:infinite-bisim}
Let $(\states, \labels, \{\kernel_\lab \mid \lab \in \labels\})$ be a NLMP and $R \subseteq \measures(\states) \times \measures(\states)$ be a symmetric relation. We say that $R$ is a (strong) \emph{probabilistic bisimulation} if for each $\mea \, R \, \meatwo$ and measurable $\Lab\subseteq\labels$
 \begin{enumerate}
  \item $\mea(\states_\Lab) = \meatwo(\states_\Lab)$, and
  \item for each $\mea \patransa{\Lab} \mea'$ there is a $\distwo \patransa{\Lab} \meatwo'$ such that $\mea' \, R \, \meatwo'$.
 \end{enumerate}
We set $\dis \bisim \distwo$ if there is a probabilistic bisimulation $R$ such that $\dis \, R \, \distwo$.
\end{definition}

\begin{example}\label{ex:cont2}
 Considering Example~\ref{ex:cont}, states $\{x\}\times[0,1]$ form a class of $\bisim$ for each $x\in[0,1]$ as the old value does not affect the behaviour. More precisely, $\dis\bisim\distwo$ iff marginals of their first component are the same.
\end{example}

\noindent\textbf{Naturalness.}
Our definition of bisimulation is not created ad-hoc as it often appears for relational definitions, but is actually an instantiation of the standard bisimulation for a particular \emph{coalgebra}. Although this aspect is not necessary for understanding the paper, it is another argument for naturalness of our definition. For reader's convenience, we present a short introduction to coalgebras and the formal definitions in \refappendix. Here we only provide an intuitive explanation by example.

Non-deterministic labelled transition systems are essentially given by the transition function $\states\to \powerset(\states)^\act$; given a state $s\in \states$ and a label $a\in\act$, we can obtain the set of the successors $\{s'\in\states\mid s\tran{a}s'\}$. The transition function corresponds to a coalgebra, which induces a bisimulation coinciding with the classical one of Park and Milner \cite{DBLP:books/daglib/0067019}. Similarly, PA are given by the transition function $\states\to \powerset(\measures(\states))^\act$; instead of successors there are distributions over successors. Again, the corresponding coalgebraic bisimulation coincides with the classical ones of Larsen and Skou~\cite{DBLP:conf/popl/LarsenS89} and Segala and Lynch \cite{DBLP:conf/concur/SegalaL94}.

In contrast, our definition can be obtained by considering states
$\states'$ to be distributions in $\measures(S)$ over the original
state space and defining the transition function to be $\states'\to
([0,1]\times\powerset(\states'))^{\Sigma(\act)}$. The difference to
the standard non-probabilistic case is twofold: firstly, we consider
all measurable sets of labels, i.e.\ all elements of $\Sigma(\act)$;
secondly, for each label set we consider the mass, i.e.\ element of
$[0,1]$, of the current state distribution that does not
deadlock\holger{reread this}, i.e.\ can perform some of the
labels. These two aspects form the crux of our approach and
distinguish it from other approaches.

%% file: applications.tex
\section{Applications}
\label{sec:applications}
\label{sec:applications!}


We now argue by some concrete application domains that the
distribution view on bisimulation yields a fruitful notion.

\subsection{Memoryless vs. memoryfull continuous time.}
\label{sec:applications-sa}

First, we reconsider the motivating discussion from
Section~\ref{sec:intro} revolving around the difference between
continuous time represented by real-valued clocks, respectively
memoryless stochastic time.
For this we introduce a simple model of \emph{stochastic
  automata}~\cite{DBLP:journals/iandc/DArgenioK05}.
\begin{definition}
 A \emph{stochastic automaton (SA)} is a tuple $\SA = (\locations,\clocks,\actions,\edges,\clocksetting, \clockdist)$ where
$\locations$ is a set of locations,
$\clocks$ is a set of clocks,
$\actions$ is a set of actions,
$\edges \;\subseteq \locations \times \actions \times 2^\clocks \times \locations$ is a set of edges,
$\clocksetting: \locations \to 2^\clocks$ is a clock setting function, and
$\clockdist$ assigns to each clock its distribution over $\Rsetpo$.
\end{definition}
Avoiding technical details, $\SA$ has the following NLMP semantics
$\pts_\SA$ with state space $\SAstates = \locations \times
\Rset^\clocks$, assuming it is initialized in some location $q_0$:
%
When a location $q$ is entered, for each clock $c \in \clocksetting(q)$ a positive \emph{value} is chosen randomly according to the distribution $\clockdist(c)$ and stored in the state space. 
Intuitively, the automaton idles in location $q$ with all all clock
values decreasing at the same speed until some edge $(q,a,\Clo,q')$
becomes \emph{enabled}, i.e. all clocks from $\Clo$ have value $\leq
0$. After this \emph{idling time} $t$, the action $a$ is taken and the
automaton enters the next location $q'$. If an edge is enabled on
entering a location, it is taken immediately, i.e. $t=0$. If more than
one edge become enabled simultaneously, one of them is chosen
non-deterministically. Its formal definition is given in~\refappendix.
%
%
%
We now are in the position to harvest
Definition~\ref{def:infinite-bisim}, to arrive at the
novel bisimulation for stochastic automata.

\begin{definition}\label{def:bisim-sta}
 We say that locations $\loc_1,\loc_2$ of an SA $\SA$ are \emph{probabilistic bisimilar}, denoted $\loc_1 \bisim \loc_2$, if $\mea_1 \bisim \mea_2$ in 
$\pts_\SA$ where each $\mea_i$ corresponds\holger{what is 'corresponds'?
  Is it the joint distribution with a marginal that is Dirac on
  $\loc_i$ and so forth?}\jena{not only this, the distribution is obtained as a product of these marginals, i.e. they are independent.}\holger{I was expecting so, but is this clear by the word 'corresponds'?}
 to the location being $\loc_i$, any $\clo \not\in \clocksetting(\loc_i)$ being $0$, and any $\clo \in \clocksetting(\loc_i)$ being independently set to a random value according to $\clockdist(\clo)$.

\end{definition}
This bisimulation identifies $\loctop$ and $\loctop'$ from
Section~\ref{sec:intro} unlike any previous bisimulation on
SA~\cite{DBLP:journals/iandc/DArgenioK05}. In
Section~\ref{sec:algorithms-infinite} we discuss how to
compute this bisimulation, despite being continuous-space. Recall that
the model initialized by $q$ is obtained by first translating two
simple CTMC, and then applying the natural interleaving semantics,
while the model, of $q'$ is obtained by first applying the equally
natural CTMC interleaving semantics prior to translation. The
bisimilarity of these two models generalizes to the whole universe of
CTMC and SA:

%
%

\begin{theorem}\label{thm:sta-commute}
  Let $SA(\mathcal{C})$ denote the stochastic automaton corresponding
  to a CTMC $\mathcal{C}$.  For any CTMC $\mathcal{C}_1,
  \mathcal{C}_2$, we have
 $$SA(\mathcal{C}_1) \parallel_{S\!A} SA(\mathcal{C}_1) 
 \;\; \bisim \;\;  
 SA(\mathcal{C}_1 \parallel_{CT} \mathcal{C}_1).$$ 
%
%
%
%
\end{theorem}
Here,  $\parallel_{CT}$ and $\parallel_{S\!A}$
denotes the interleaving parallel composition of
SA~\cite{DBLP:journals/iandc/DArgenioK05a} (echoing TA parallel
composition) and CTMC~\cite{pepa,HHM} (Kronecker sum of their matrix
representations), respectively. 
%

\holger{dropped poulation model discussion. Do not understand it. It
  smells like you are discussing the (deterministic) ``mean field'' of
  an infinite population Markov models. Do you?}  \todo{Jan: basically
  yes. have they defined bisimulation for that? depends how the
  deterministic system is defined. what do they do when some processes
  cannot execute the action? (die? wait? something else?)}
\holger{Notably, in the mean-field world I know of, the population
  models in this context are CT(!)MCs, and actions, resource
  contention, nondeterminism and the like are difficult to formulate,
  unless you move out of the CTMC realm....}
\jena{Yes, mean-field is continuous-time. However, the references I
  digged out talk about discrete-time models. Mentioning shorter in
  Further applications is perfect for me.}\holger{Am unable to do it,
  because I do not understand.}

\subsection{Bisimulation for partial-observation MDP (POMDP).} 

A POMDP is a quadruple $\mathcal{M}=(\states,\labels,\delta, \mathcal{O})$ where (as in an MDP) $\states$ is a set of states, $A$ is a set of actions, and $\delta: \states \times \actions \to \dist(\states)$ is a transition function. Furthermore, $\mathcal{O} \subseteq 2^\states$ partitions the state space. The choice of actions is resolved by a policy yielding a Markov chain. Unlike in an MDP, such choice is not based on the knowledge of the current state, only on knowing that the current state belongs into an \emph{observation} $o \in \mathcal{O}$.
 POMDPs have a wide range of applications in robotic control, automated planning, dialogue systems, medical diagnosis, and many other areas~\cite{DBLP:journals/aamas/ShaniPK13}.

In the analysis of POMDP, the distributions over states, called \emph{beliefs}, arise naturally. They allow for transforming the POMDP $\mathcal{M}$ into a fully observable NLMP $\PA_{\mathcal{M}} = (\states,\mathcal{O}, \patrans)$ with continuous space, by setting $(\sta, \patransa{o},\dis) \in \patrans$ if $\sta \in o$ and $\delta(\sta,a) = \dis$ for some $a\in A$. Although probabilistic bisimulations over beliefs have been already considered~\cite{DBLP:conf/ijcai/CastroPP09,DBLP:conf/nfm/JansenNZ12}, no connection of this particular case to general probabilistic bisimulation has been studied. 
We can set $\dis \bisim \dis'$ in $\mathcal{M}$ if $\dis \bisim \dis'$ in $\PA_\mathcal{M}$.
In Section~\ref{sec:algorithms-finite}, we shall provide an algorithm for computing bisimulations over beliefs in finite POMDP. Previously, there was only an algorithm~\cite{DBLP:conf/nfm/JansenNZ12} for computing bisimulations on distributions of Markov \emph{chains} with partial observation. 

\subsection{Further applications.}
Probabilistic automata are especially apt for compositional modelling
of \emph{distributed systems}.  The only information a component in a
distributed system has about the current state of another component
stems from their mutual communication. Therefore, each component can
be also viewed from the outside as a partial-observation system.  Thus, also in this
context, distribution bisimulation is a natural concept. 

Furthermore we can  understand a PA as a description, in the sense
of~\cite{DBLP:journals/deds/GastG11,may1974biological}, 
 of a representative
\emph{agent} in a large homogeneous \emph{population}.
The distribution view then naturally represents the ratios of agents being currently in the individual
states and labels given to this large population of PAs correspond to global control actions~\cite{DBLP:journals/deds/GastG11}.
For more details on applications, see~\refappendix.

\todo{explain more}

\todo{mention the population model thougts.}


%% file: algorithms.tex
\section{Algorithms}
\label{sec:algorithms}

In this section, we discuss computational aspects of deciding our
bisimulation. Since $\bisim$ is a relation over distributions over the
system's state space, it is uncountably infinite even for simple
finite systems, which makes it in principle intricate to
decide. Fortunately, the bisimulation relation has a linear structure,
and this allows us to employ methods of linear algebra to work with it
effectively. Moreover, important classes of continuous-space
systems can be dealt with, since their structure can be exploited. We
exemplify this on a subset of deterministic stochastic
automata, for which we are able to provide an algorithm to decide
bisimilarity.

\subsection{Finite systems -- greatest fixpoints}
\label{sec:algorithms-finite}

Let us fix a PA $(\states, \labels, \patrans)$. We apply the standard approach by starting with $\dist(\states) \times \dist(\states)$ and pruning the relation until we reach the fixpoint $\bisim$.  
In order to represent $\bisim$ using linear algebra, we identify a distribution $\dis$ with a vector $(\dis(\sta_1),\ldots,\dis(\sta_{|\states|}))\in\Rset^{|\states|}$. 

Although the space of distributions is uncountable,
we construct an implicit representation of $\bisim$ by a system of equations written as columns in a matrix $\B$.
\begin{definition}
 A matrix $\B$ with $|\states|$ rows is a \emph{bisimulation matrix} if 
 for some bisimulation $R$, for any distributions $\dis,\distwo$  
 $$\dis\, R\,\distwo \;\;\;\text{iff}\;\;\; (\dis-\distwo)\B=0.$$
\end{definition}

\noindent
For a bisimulation matrix $\B$, an equivalence class of $\dis$ is then the set $(\dis+\{\rho\mid\rho\B=0\})\cap\dist(\states)$, the set of distributions that are equal modulo $\B$.

\begin{example}\label{ex:bisim-matrix}
The bisimulation matrix $\B$ below encodes that several conditions must hold for two distributions $\dis,\distwo$ to be bisimilar. 
Among others, if we multiply $\dis-\distwo$ with e.g.\ the second column, we must get $0$.
This translates to $(\dis(v)-\distwo(v))\cdot1=0$, i.e. $\dis(v)=\distwo(v)$.
Hence for bisimilar distributions, the measure of $v$ has to be the same. 
This proves that $u\not\bisim v$ (here we identify states and their Dirac distributions).
Similarly, we can prove that $\;t \;\bisim\; \frac{1}{2} t' + \frac{1}{2} t''$. Indeed, if we multiply the corresponding difference vector $(0,0,1,-\frac12,-\frac12,0,0)$ with any column of the matrix, we obtain $0$.

%

\begin{tikzpicture}[outer sep=0.1em,->, xscale=0.9, yscale=1.2,
state/.style={draw,circle,minimum size=1.6em,inner sep=0.1em}]
\begin{scope}
\node[state] (s) at (0,0.5) {$s$};
\node[state] (t) at (1,0.5) {$t$};
\node[state] (u) at (2,1) {$u$};
\node[state] (v) at (2,0) {$v$};
\path[->] 
(s) edge node[above] {$a$} (t)
(t) edge node[above]{\textonehalf} node[below]{$a$}(u)
(t) edge node[below]{\textonehalf} (v)
(u) edge[loop above,looseness=4] node[right]{$b$} (u)
(v) edge[loop above,looseness=4] node[right]{$c$} (v)
;
\node[state] (s') at (4,0.5) {$s'$};
\node[state] (t1) at (3,1) {$t'$};
\node[state] (t2) at (3,0) {$t''$};
\path[->] 
(s') edge node[below] {$a$} node[above]{\textonehalf} (t1)
(s') edge  node[below]{\textonehalf} (t2)
(t1) edge node[below] {$a$} (u)
(t2) edge node[below] {$a$} (v)
;
\end{scope}
\begin{scope}[xshift=20em,yshift=16,scale=0.8, every node/.style={scale=0.8},]
\node [font=\small,gray] at(-2.3,.05){
$\begin{array}{c}
  {s:} \\
  {s':} \\
  {\footnotesize t:} \\
  {\footnotesize t':} \\
  {\footnotesize t'':} \\
  {\footnotesize u:} \\
  {\footnotesize v:}
 \end{array}$};
 
 \node [font=\small] at(0,0){
$\left(\begin{array}{ccccc}
1&0&0&0&0\\
1&0&0&0&0\\
1&0&0&\text{\textonehalf}&\text{\textonehalf}\\
1&0&0&0&1\\
1&0&0&1&0\\
1&0&1&0&0\\
1&1&0&0&0
\end{array}\right)$};
\end{scope}
\end{tikzpicture}
\end{example}

Note that the unit matrix is always a bisimulation matrix, not relating anything with anything but itself. For which bisimulations do there exist bisimulation matrices?
We say a relation $R$ over distributions is \emph{linear} if $\dis R \distwo$ and $\dis' R \distwo'$ imply $\big(p\dis+(1-p)\dis'\big)\;R\;\big( p\distwo+(1-p)\distwo'\big)$ for any $p\in[0,1]$.

\begin{lemma}\label{lem:existence}
For every linear bisimulation there exists a corresponding bisimulation matrix.
\end{lemma}
\jena{no proof sketches? I would then state it as one lemma and give this splitting into two parts as a two-line sketch}
Since $\bisim$ is linear (see \refappendix), there is a bisimulation matrix corresponding to $\bisim$. 
It is a least restrictive bisimulation matrix $\B$ (note that
all bisimulation matrices with the least possible dimension have identical solution space), we call it \emph{minimal bisimulation matrix}.
We show that the necessary and sufficient condition for $E$ to be a bisimulation matrix is \emph{stability} with respect to transitions. 

\begin{definition}
For a $|\states|\times|\states|$ matrix $P$, we say that a matrix $\B$ with $|\states|$ rows is \emph{$P$-stable} if for every $\rho\in\Rset^{|\states|}$, 
\begin{align}
\rho \B=0 \implies \rho P \B=0 
\end{align}
\end{definition}

\noindent
We first briefly explain the stability in a simpler setting.

\subsubsection{Action-deterministic systems.}  
%
%
Let us consider PA where in each state, there is at most one transition. For each $a\in\act$, we let $P_a=(p_{ij})$ denote the transition matrix such that for all $i,j$, if there is (unique) transition $\sta_i\patransa{a} \mu$ we set $p_{ij}$ to $\mu(\sta_j)$, otherwise to $0$. Then $\dis$ evolves under $a$ into $\dis P_a$. Denote $\vec{1}=(1,\ldots,1)^\top$.
\begin{proposition}\label{prop-alg-determ}
In an action-deterministic PA, $\B$ containing $\vec 1$ is a bisimulation matrix iff it is  $P_\lab$-stable for all $\lab\in\labels$.
\end{proposition}
To get a minimal bisimulation matrix $\B$, we start with a single vector $\vec{1}$ which stands for an equation saying that the overall probability mass in bisimilar distributions is the same. Then we repetitively multiply all vectors we have by all the matrices $P_a$ and add each resulting vector to the collection if it is linearly independent of the current collection, until there are no changes.
In Example \ref{ex:bisim-matrix}, the second column of $\B$ is obtained as $P_c\vec1$, the fourth one as $P_a(P_c\vec1)$ and so on.

The set of all columns of $\B$ is thus given by the described iteration $$\{P_a\mid a\in\act\}^*\vec{1}$$
modulo linear dependency.  Since $P_a$ have $|\states|$ rows, the
fixpoint is reached within $|\states|$ iterations yielding $1\leq d\leq|\states|$
equations\jena{note that \# of iters does not equal \# of equations}. Each class then forms an $(|\states|-d)$-dimensional
affine subspace intersected with the set of probability distributions
$\dist(\states)$. This is also the principle idea behind the algorithm of~\cite{DBLP:journals/siamcomp/Tzeng92} and~\cite{DBLP:journals/ijfcs/DoyenHR08}.
%

\subsubsection{Non-deterministic systems.}

\newcommand{\PaW}{P_A^W}
\newcommand{\PaWc}{P_A^{W(c)}}
\newcommand{\bigC}{C}
\newcommand{\rohy}{\corners(\bigC)}

In general, for transitions under $A$, we have to consider $c_i^A$ non-deterministic choices in each $\sta_i$ among all the outgoing transitions under some $a\in A$.
We use variables $w_i^j$ denoting the probability that $j$-th transition, say $(\sta_i,a_i^j,\dis_i^j)$, is taken by the scheduler/player\footnote{
We use the standard notion 
of Spoiler-Duplicator bisimulation game (see e.g.~\cite{Sangiorgi:2011:ATB:2103601})
where in $\{\mu_0,\mu_1\}$ Spoiler chooses $i\in\{0,1\},A\subseteq\act$, and $\dis_i \tran{A} \dis_i'$, Duplicator has to reply with $\dis_{1-i}\tran{A}\dis_{1-i}'$ such that $\dis_i(\states_A) = \dis_{i-1}(\states_A)$, and the game continues in $\{\mu_0',\mu_1'\}$. Spoiler wins iff at some point Duplicator cannot reply.
}
in $\sta_i$. We sum up the choices into a ``non-deterministic'' transition matrix $\PaW$ with parameters $W$ whose $i$th row equals $\sum_{j=1}^{c_i^A} w^j_i \dis_i^j$. It describes where the probability mass moves from $\sta_i$ under $A$ depending on the collection $W$ of the probabilities the player gives each choice. By $\choices_A$ we denote the set of all such $W$.
%
%
%
%

A simple generalization of the approach above would be to consider
$\{P_A^W\mid A\subseteq\act, W\in\choices_A\}^* \vec{1}$. However,
firstly, the set of these matrices is uncountable whenever there are
at least two transitions to choose from. Secondly, not all $P_A^W$ may be used as the following example shows.

\begin{example}
In each bisimulation class in the following example, the probabilities
of $s_1 + s_2$, $s_3$, and $s_4$ are constant, as can also be seen from the bisimulation matrix $\B$, similarly to Example \ref{ex:bisim-matrix}.
Further, $\B$ can be obtained as $(\vec{1}\;\, P_c \vec{1} \;\, P_b \vec{1})$. 
Observe that $\B$ is $P_{\{a\}}^{W}$-stable for $W$ that maximizes the probability of going into the ``class'' $s_3$ (both $s_1$ and $s_2$ go to $s_3$, i.e. $w_1^1 = w_2^1 = 1$); similarly for the ``class''~$s_4$.
%
%

\noindent
\begin{center}
\begin{tikzpicture}[x=2.5cm,y=1.2cm,outer sep=1mm,yscale=0.8,
state/.style={circle,draw,minimum size=1.6em,inner sep=0.1em},
trans/.style={font=\scriptsize,->}]
\node[state] (s) at (0,0) {$\sta_1$};
\node[state] (t) at (0,-1) {$\sta_2$};
\node[state] (u) at (0.6,0) {$\sta_3$};
\node[state] (v) at (0.6,-1) {$\sta_4$};
\path[trans] (s) edge node[pos=0.3,trans,above=-2]{$a$} (u);
\path[trans] (s) edge node[pos=0.3,trans,above=-2]{$a$} (v);
\path[trans] (t) edge node[pos=0.3,trans,below]{$a$} (u);
\path[trans] (t) edge node[pos=0.3,trans,below]{$a$} (v);
\path[trans] (u) edge[loop right,looseness=5] node[trans,below]{$b$} (u);
\path[trans] (v) edge[loop right,looseness=5] node[trans,above]{$c$} (v);
\begin{scope}[scale=0.8, every node/.style={scale=0.8},]
\node at(2.5,-0.7){
$P_{\{a\}}^W=\left(\begin{array}{cccc}
0&0&w_1^1&w_2^2\\
0&0&w_2^1&w_2^2\\
0&0&0&0\\
0&0&0&0
\end{array}\right)$};

\node at(4.4,-0.7){
$E=\left(\begin{array}{ccc}
1&0&0\\
1&0&0\\
1&0&1\\
1&1&0
\end{array}\right)$};

\end{scope}
\end{tikzpicture}
\end{center}

However, for $W$ with $w_1^1\neq w_2^1$, e.g.\ $s_1$ goes to $s_3$ and $s_2$ goes with equal probability to $s_3$ and $s_4$ ($w_1^1=1, w_2^1=w_2^2 = \frac{1}{2}$), we obtain from $P_{\{a\}}^{W}E$ a new independent vector $(0,0.5,0,0)^\top$ enforcing a partition finer than $\bisim$.
%
This does not mean that Spoiler wins the game when choosing such mixed $W$ in some $\mu$, it only means that Duplicator needs to choose a \emph{different} $W'$ in a bisimilar $\nu$ in order to have $\mu P_A^W \bisim \nu P_A^{W'}$ for the successors.

\end{example}

A fundamental observation is that we get the correct bisimulation when Spoiler is restricted to finitely many ``extremal'' choices and 
Duplicator is restricted for such extremal $W$ to respond only with the very same $W$.
%
%
%

\jena{from here on extremely technical, at least correct hopefully}
To this end, consider $M_A^W=P_A^W \B$ where $\B$ is the current matrix with each of $e$ columns representing an equation. Intuitively, the $i$th row of $M_A^W$ describes how much of $\sta_i$ is moved to various classes when a step is taken. Denote the linear forms in $M_A^W$ over $W$ by $m_{ij}$.
Since the players can randomize and mix choices which transition to take, the set of vectors $\{(m_{i1}(w_i^1,\ldots,w_i^{c_i}),\ldots,m_{ib}(w_i^1,\ldots,w_i^{c_i}))\mid w_i^1,\ldots,w_i^{c_i}\geq 0,\sum_{j=1}^{c_i}w_i^j=1\}$ forms a convex polytope denoted by $C_i$. Each vector in $C_i$ is thus the $i$th row of the matrix $M_A^W$ where some concrete weights $w_i^j$ are ``plugged in''. This way $C_i$ describes all the possible choices in $\sta_i$ and their effect on where the probability mass is moved.

Denote vertices (extremal points) of a convex polytope $P$ by $\corners(P)$. Then $\corners(C_i)$ correspond to pure (non-randomizing) choices that are ``extremal'' w.r.t.~$\B$. Note that now if $\sta_j\sim \sta_k$ then $C_j=C_k$, or equivalently $\corners(C_j)=\corners(C_k)$. Indeed, for every choice in $\sta_j$ there needs to be a matching choice in $\sta_k$ and vice versa.
However, since we consider bisimulation between generally non-Dirac distributions, we need to combine these 
extremal choices. We define the set $\corners(\bigC)\subseteq\prod_{i=1}^{|\states|}\corners(C_i)$ to contain a tuple\jena{sounds like convex combinations, puzzled me a bit} $c=(c_1\ \cdots\ c_{|\states|})$ iff the $c_i$'s are ``extremal in (some) same direction'', i.e. 
$\sum_{i=1}^{|\states|} c_i$ is a vertex (extremal choice) of the polytope generated by points $\{\sum_{i=1}^{|\states|}c_i'\mid\forall i: c_i'\in C_i\}$.
Each $c\in\corners(\bigC)$ is a tuple of vertices, 
and thus corresponds to particular choices, denoted by $W(c)$. 
 \jena{would prefer to have the algorithm top aligned on the page; otherwise, put it between Prop 2 and Thm 2. Here it interrupts the flow of the explanation.}

\begin{proposition}\label{prop-alg-nondeterm}
Let $\B$ be a matrix containing $\vec 1$. It is a bisimulation matrix iff it is $\PaWc$-stable for all $A\subseteq\act$ and $c\in\corners(\bigC)$.
\end{proposition}

\begin{algorithm}[ht]
\myspace
\SetAlgoLined
\DontPrintSemicolon
\SetKwInOut{Parameter}{parameter}\SetKwInOut{Input}{Input}\SetKwInOut{Output}{Output}
\SetKwData{C}{C}\SetKwData{D}{D}\SetKwData{MX}{M}\SetKwData{f}{f}

\Input{Probabilistic automaton $(\states,\act,\tran{})$}
\Output{A minimal bisimulation matrix $E$}
\BlankLine
\ForEach{$A\subseteq\act$}
{compute $\PaW$ \hfill\texttt{// }\textsf{non-deterministic transition matrix~~~~~}}
$\B\gets(\vec{1})$\;
\Repeat{$\B$ does not change}
{
\ForEach{$A\subseteq\act$}
{
$M_A^W\gets \PaW\B$ \hfill~~~~~~\texttt{// }\textsf{polytope of all choices}\;
compute $\corners(\bigC)$ from $M_A^W$ \hfill\texttt{// }\textsf{vertices, i.e. extremal choices}\;
\ForEach{$c\in\corners(\bigC)$}{
$M_A^{W(c)}\gets M_A^W$ with values $W(c)$ plugged in\;
$\B_{new}\gets$columns of $M_A^{W(c)}$ linearly independent of columns of $\B$\;
$\B\gets(\B\ \B_{new})$\;
}
}
}
\BlankLine
\caption{Bisimulation on probabilistic automata}
\label{alg-fin}
\end{algorithm}

\begin{theorem}\label{thm:algorithm-finite}
Algorithm~\ref{alg-fin} computes a minimal bisimulation matrix.
\end{theorem}
%
%

The running time is exponential. We leave the question whether linear programming or other methods \cite{DBLP:conf/fsttcs/HermannsT12} can yield $\B$ in polynomial time open. The algorithm can easily be turned into one computing other bisimulation notions from the literature, for which there were no algorithms so far, see Section \ref{sec:rw}.

%

\subsection{Continuous-time systems - least fixpoints}
\label{sec:algorithms-infinite}

Turning our attention to continuous systems, we finally sketch an
algorithm for deciding bisimulation $\bisim$ over a subclass of
stochastic automata, this constitutes the first algorithm to compute
a bisimulation on the uncountably large semantical object.
%
%

%
We need to adopt two restrictions. 
%
First, we consider only \emph{deterministic} SA, where the probability that two edges  become enabled at the same time is zero (when initiated in any location). 
Second, to simplify the exposition, we restrict all distributions
occurring to exponential distributions. Notably, even for this class,
our bisimulation is strictly coarser than the one induced by standard
bisimulations~\cite{pepa,HHM,ymca} for continuous-time Markov
chains. At the end of the section we discuss possibilities for
extending the class of supported distributions.
%
%
%
Both the restrictions can be effectively checked on SA.
%

\begin{theorem}\label{thm:tableau}
Let $\SA = (\locations,\clocks,\actions,\edges,\clocksetting, \clockdist)$ be a deterministic SA over exponential distributions.
There is an algorithm to decide
in time polynomial in $|\SA|$ and exponential in $|\clocks|$ whether $\loc_1 \bisim \loc_2$ 
for any locations $\loc_1,\loc_2$.
\end{theorem}

The rest of the section deals with the proof.\holger{I can broadly
  follow what is happening here.} We fix $\SA = (\locations,\clocks,\actions,\edges,\clocksetting, \clockdist)$ and $q_1,q_2 \in \locations$.
%
%
First, we straightforwardly abstract the NLMP semantics $\pts_\SA$ by a NLMP $\abssemantics$
over state space $\absstates = \locations \times (\Rsetpo \cup
\{-\})^\clocks$ where all negative values of clocks are expressed by
one element $-$. Let $\xi$ denote the obvious mapping of distributions
$\dist(\SAstates)$ onto $\dist(\absstates)$. Then $\xi$
preserves bisimulation since two states $s_1,s_2$ that differ only in negative values satisfy $\xi(\kernel_a(s_1)) = \xi(\kernel_a(s_2))$ for all $a\in\labels$.

\begin{lemma}\label{lem:abs}
For any distributions $\dis, \distwo$ on $\SAstates$ we have $\dis \bisim \distwo$ iff $\xi(\dis) \bisim \xi(\distwo)$.
\end{lemma}

Second, similarly to an embedded Markov chain of a CTMC, we further abstract the NLMP $\abssemantics$ by a \emph{finite} deterministic PA
$\chain = (\chainstates, \actions, \patrans)$ such that each state of $\chain$ is a distribution over the uncountable state space $\absstates$.
\begin{itemize}
\item The set $\chainstates$ is the set of states reachable via the transitions relation defined below from the distributions $\dis_1, \dis_2$ corresponding to $\loc_1$, $\loc_2$ (see Definition~\ref{def:bisim-sta}).
\item Let us fix a state $\dis \in \chainstates$ (note that $\dis\in\measures(\absstates)$) and an action $a\in\actions$ such that in the NLMP $\abssemantics$ an $a$-transition occurs with positive probability, i.e. $\dis \patransa{\Lab_a} \distwo$ for some $\distwo$ and for $\Lab_a = \{a\}\times \Rsetpo$. Thanks to restricting to deterministic SA, $\abssemantics$ is also deterministic and such a distribution $\distwo$ is uniquely defined. 
We set $(\dis,a,M) \in \; \patrans$ where $M$ is the discrete distribution that assigns probability $p_{\loc,f}$ to state $\distwo_{\loc,f}$ for each $\loc\in\locations$ and $f: \clocks \to \{-,+\}$ where $p_{\loc,f} = \distwo(\absstates_{\loc,f})$, $\distwo_{\loc,f}$ is the conditional distribution $\distwo_\loc(X) := \distwo(X \cap \absstates_{\loc,f})/\distwo(\absstates_{\loc,f})$ for any measurable $X \subseteq \absstates$, and $\absstates_{\loc,f} = \{(\loc',v)\in\absstates \mid \loc' = \loc, \text{$v(\clo) \geq 0$ iff $f(\clo) = +$ for each $\clo\in\clocks$}\}$ the set of states with location $\loc$ and where the sign of clock values matches $f$.
\end{itemize}
For exponential distributions all the reachable states $\nu \in \chainstates$ correspond to some location $\loc$ where the subset $\Clo\subseteq \clocks$ is newly sampled, hence we obtain:

\begin{lemma}\label{lem:expo-finite}
For a deterministic SA over exponential distributions, $|\chainstates| \leq |\locations|2^{|\clocks|}$.
\end{lemma}

%
%
%
%
%
Instead of a greatest fixpoint computation as employed for the discrete
algorithm
, we take a complementary approach and prove or disprove
bisimilarity by a least fixpoint procedure.
We start with the initial pair of distributions (states in $\chain$)
which generates further requirements that we impose on the relation and try to satisfy them. 
We work with a \emph{tableau}, a rooted tree where each node is either an \emph{inner node} with a pair of discrete probability distributions over states of $\chain$ as a label, a \emph{repeated node} with a label that already appears somewhere between the node and the root, or a \emph{failure node} denoted by $\fail$, and the children of each inner node are obtained by one \emph{rule} from $\{\tableau{Step}, \tableau{Lin} \}$. A tableau not containing $\fail$ is \emph{successful}.
%
\begin{description}
 \item[\tableau{Step}] For a node $\mu \sim \nu$ where $\mu$ and $\nu$ have \emph{compatible timing}, we add for each label $\lab\in\labels$ one child node $\mu_a \sim \nu_a$ where $\mu_a$ and $\nu_a$ are the unique distributions such that $\mu \patransa{a} \mu_a$ and $\nu \patransa{a} \nu_a$. Otherwise, we add one failure node. 
 We say that $\mu$ and $\nu$ have compatible timing if for all actions $a\in\actions$ we have $\mu(\states_{\Lab_a}) = \nu(\states_{\Lab_a})$ and 
 if for all actions $a\in\actions$ with $\mu(\states_{\Lab_a}) >0$ we
 have that  
 $\mu$ restricted to $\states_{\Lab_a}$ is equivalent to $\nu$ restricted to $\states_{\Lab_a}$.
\item[\tableau{Lin}] For a node $\mu \bisim \nu$ linearly dependent on the set of 
remaining nodes in the tableau, we add one child (repeat) node $\mu \bisim \nu$. Here, we understand each node $\mu \bisim \nu$ as a vector $\mu - \nu$ in the $|\states_\SA|$-dimensional vector space. 
\end{description} 

%
\noindent
Note that compatibility of timing is easy to check. Furthermore, the set of rules is correct and complete w.r.t. bisimulation in $\abssemantics$.
\begin{lemma}\label{prop:correctness}
There is a successful tableau from $\dis \bisim \distwo$ iff $\dis \bisim \distwo$ in $\abssemantics$.
Moreover, the set of nodes of a successful tableau is a subset of a bisimulation. 
\end{lemma}

\noindent
We get Theorem~\ref{thm:tableau} since $q_1 \bisim q_2$ iff $\xi(\dis_1) \bisim \xi(\dis_2)$ in $\abssemantics$ and since, thanks to \tableau{Lin}:

\begin{lemma}\label{lem:finite-tableau}
There is a successful tableau from $\dis \bisim \distwo$ iff there is a finite successful tableau from $\dis \bisim \distwo$ of size polynomial in $|\chainstates|$.
\end{lemma}
%

%

%
%

\begin{example}\label{ex:tableau}
Let us demonstrate the rules by a simple example.
Consider the following stochastic automaton $\SA$ on the left.
\vspace{-0.7em}
\begin{center}
\begin{tikzpicture}[x=2.5cm,y=1.2cm,outer sep=1mm,
state/.style={draw,circle, inner sep =0.4em, text centered},
trans/.style={font=\scriptsize},
prob/.style={font=\scriptsize}
]

\begin{scope}
\node[state] (s) at (0,0) {$\loctop$};
\node[state] (t) at (0.8,0) {$\locleft$};
\node[state] (u) at (1.7,0) {$\locright$};
\node[above=0,prob] at (s.north) {$x:=\mathrm{Exp}(1/2)$}; 
\node[above=-8,prob] at (s.north) {$y:=\mathrm{Exp}(1/2)$}; 
\node[above=-8,prob] at (t.north) {$x:=\mathrm{Exp}(1)$}; 
\node[above=-8,prob] at (u.north) {$x:=\mathrm{Exp}(1)$};

\path[->] (s) edge[loop left,in=205,out=165,looseness=5] 
node[prob,below=-2,pos=0.7]{$x = 0$} node[prob,below=-3,pos=0.1]{$a$} (s);
\path[->] (s) edge[] node[prob,above=-4] {$a$} node[prob,below=-4] {$y=0$} (t);
\path[->] (t) edge[loop right,in=25,out=345,looseness=5] 
node[prob,below=-2,pos=0.3]{$x = 0$} node[prob,above=-3,pos=0.1]{$a$} (t);
\path[->] (u) edge[loop right,in=25,out=345,looseness=5] 
node[prob,below=-2,pos=0.3]{$x = 0$} node[prob,above=-3,pos=0.1]{$a$} (u);

\draw[dotted,thick] (2.4,0.5) -- (2.4,-0.3);
\end{scope}

\begin{scope}[xshift=200]
\node[state,inner sep =0.2em] (s) at (0,0) {$\dis_q$};
\node[state,inner sep =0.2em] (t) at (0.75,0) {$\dis_u$};
\node[state,inner sep =0.2em] (u) at (1.25,0) {$\dis_v$};
\node[above left=-8,prob] at (s.north west) {}; \node[above right=-8,prob] at (t.north) {}; \node[above right=-8,prob] at (u.north) {};

\path[->] (s) edge[loop,in=60,out=0,looseness=5] node[above left=-4,pos=0.1,prob]{$a$} node[pos=0.2,name=x,inner sep=0,outer sep=0]{} node[above=-3,pos=0.8,prob] {$0.5$} (s);
\path[->] (x) edge[bend left] node[prob,above=-4] {$0.5$} (t);
\path[->] (t) edge[loop above,looseness=5] 
node[prob,right=-4,pos=0.1]{$a$} (t);
\path[->] (u) edge[loop above,looseness=5] 
node[prob,right=-4,pos=0.1]{$a$} (u);
\end{scope}

%
%
\end{tikzpicture}
\end{center}
\vspace{-0.7em}
Thanks to the exponential distributions, $\chain$ on the right has also only three states where $\dis_q = q \otimes Exp(1/2) \otimes Exp(1/2)$ is the product of two exponential distributions with rate $1/2$, $\dis_u = u\otimes Exp(1)$, and $\dis_v = v \otimes Exp(1)$. Note that for both clocks $x$ and $y$, the probability of getting to zero first is $0.5$. 

%

\vspace{0.2em}
\begin{center}
\begin{tikzpicture}[thin,scale=0.8, every node/.style={scale=0.8}]
\begin{scope}[xshift=0em,yshift=-2em]
 \node (a1) at (0,0) {$
1 \cdot \dis_\locleft
\; \sim \;
1 \cdot \dis_\locright
$};
\draw (a1.south east) -- (a1.south west);
\node [label,right] at (a1.south east) {\tableau{Step}};

 \node (a2) at (0,-0.7) {$
1 \cdot \dis_\locleft
\; \sim \;
1 \cdot \dis_\locright
$};
\end{scope}

\begin{scope}[xshift=20em]
 \node (a1) at (0,0) {$
1 \cdot \dis_\loctop + 0 \cdot \dis_\locleft
\; \sim \;
1 \cdot \dis_\locright
$};

 \node (a2) at (0,-0.7) {$
\frac{1}{2} \cdot \dis_\loctop + \frac{1}{2} \cdot \dis_\locleft
\; \sim \;
1 \cdot \dis_\locright
$};
\node (a3) at (0,-1.4) {$
\frac{1}{4} \cdot \dis_\loctop + \frac{3}{4} \cdot \dis_\locleft
\; \sim \;
1 \cdot \dis_\locright
$};
\node (a4) at (0,-2.1) {$\cdots
$};

\draw (a1.south east) -- (a1.south west);
\node [label,right] at (a1.south east) {\tableau{Step}};
\draw (a2.south east) -- (a2.south west);
\node [label,right] at (a2.south east) {\tableau{Step}};
\draw (a3.south east) -- (a3.south west);
\node [label,right] at (a3.south east) {\tableau{Step}};

\end{scope}
\end{tikzpicture}
\end{center}
\vspace{-1em}
\noindent
The finite tableau on the left is successful since it ends in a
repeated node, thus it proves $\locleft \sim \locright$. The infinite tableau on the right is also successful and proves $\loctop\sim \locright$. When using only the rule \tableau{Step}, it is necessarily infinite as no node ever repeats.
The rule \tableau{Lin} provides the means to truncate such infinite sequences.
%
%
Observe that the third node in the tableau on the right above is
linearly dependent on its ancestors.
\end{example}
\begin{remark}
  Our approach can be turned into a complete proof system for
  bisimulation on models with \emph{expolynomial}
  distributions~\footnote{With density that is positive on an interval
    $[\ell,u)$ for $\ell \in \Nseto$, $u \in \Nset \cup \{\infty\}$
    given piecewise by expressions of the form $\sum_{i=0}^I
    \sum_{j=0}^J a_{ij} x^i e^{-\lambda_{ij}x}$ for
    $a_{ij},\lambda_{ij} \in \Rset \cup \{\infty\}$. This class
    contains many important distributions such as exponential, or
    uniform, and enables efficient approximation of others.}. Thanks
  to their properties, the states of the discrete transition system
  $\chain$ can be expressed symbolically. In fact, we conjecture that
  the resulting semi-algorithm can be twisted to a decision algorithm
  for this expressive class of models.
Being technically demanding, it is out of scope of this paper.
\end{remark}

%% file: rw.tex
\section{Related work and discussion}\label{sec:rw}
For an overview of coalgebraic work on probabilistic bisimulations we refer to a survey \cite{DBLP:journals/tcs/Sokolova11}. A considerable effort has been spent to extend this work to continuous-space systems: the solution of \cite{DBLP:conf/icalp/VinkR97} (unfortunately not applicable to $\Rset$), the construction of \cite{DBLP:journals/mscs/Edalat99} (described by \cite{Sangiorgi:2011:ATB:2103601} as ``ingenious and intricate''), sophisticated measurable selection techniques in \cite{DBLP:conf/icalp/Doberkat03}, and further approaches of \cite{DBLP:conf/lics/DesharnaisGJP00} or \cite{phd-cont}. In contrast to this standard setting where relations between states and their successor distributions must be handled, our work uses directly relations on distributions 
which simplifies the setting. The coalgebraic approach has also been applied to trace semantics of uncountable systems \cite{DBLP:conf/concur/KerstanK12}. Coalgebraic treatment of probabilistic bisimulation is still very lively \cite{DBLP:conf/fossacs/Mio14}.

Recently, distribution-based bisimulations have been studied. In \cite{DBLP:journals/ijfcs/DoyenHR08}, a bisimulation is defined in the context of language equivalence of Rabin's deterministic probabilistic automata and also an algorithm to compute the bisimulation on them. However, only finite systems with no non-determinism are considered. The most related to our notion are the very recent independently developed \cite{DBLP:journals/corr/FengZ13} and \cite{DBLP:journals/corr/abs-1202-4116}. However, none of them is applicable in the continuous setting and for neither of the two any algorithm has previously been given. Nevertheless, since they are close to our definition, 
our algorithm with only small changes can actually compute them. 
Although the bisimulation of \cite{DBLP:journals/corr/FengZ13} in a rather complex way extends \cite{DBLP:journals/ijfcs/DoyenHR08} to the non-deterministic case 
reusing 
their notions, it can be equivalently rephrased as our
Definition~\ref{def:infinite-bisim} only considering singleton sets
$A\subseteq\act$. Therefore, it is sufficient to only consider
matrices $\PaW$ for singletons $A$ in our algorithm. Apart from being
a weak relation\holger{note this}, the bisimulation  of
\cite{DBLP:journals/corr/abs-1202-4116} differs in the definition of
$\dis\tran{A}\distwo$: instead of restricting to the states of the
support that can perform \emph{some} action of $A$, it considers those
states that can perform \emph{exactly} actions of $A$. Here each $i$th
row of each transition matrix $\PaW$ needs to be set to zero if the set of labels from $s_i$ is different from $A$.

There are also bisimulation relations over distributions that, however, coincide with  the classical \cite{DBLP:conf/popl/LarsenS89} on Dirac distributions and are only directly lifted to non-Dirac distributions.  
Thus they fail to address the motivating correspondence problem from Section~\ref{sec:intro} and are less precise for large-population models. Moreover, no
algorithms were given. 
They were considered for finite \cite{DBLP:conf/concur/CrafaR11,DBLP:dblp_journals/fac/Hennessy12} and uncountable \cite{cattani-thesis} state spaces.

There are other bisimulations that identify more states than the classical \cite{DBLP:conf/popl/LarsenS89} such as \cite{DBLP:conf/concur/SongZG11} and \cite{italie} designed to match a specific logic. Further, weak bisimulations coarser than usual state based analogues were given in \cite{DBLP:conf/lics/EisentrautHZ10,DBLP:conf/qest/EisentrautHKT013,DBLP:journals/iandc/DengH13}, which also inspires our work, especially their approach to internal transitions. However, they are quite different from our notion as in the case without internal transitions they basically coincide with lifting
\cite{DBLP:dblp_journals/fac/Hennessy12} of the classical bisimulation \cite{DBLP:conf/popl/LarsenS89}. Another approach to obtain coarser equivalences on probabilistic automata is via testing scenarios~\cite{DBLP:conf/icalp/StoelingaV03}.


%% file: app.tex
\newcommand{\lijunbisim}{\bisim_{\functorfin}}

\section{Bisimulation coalgebraically}

\subsection{Short introduction to coalgebras}

Definitions of bisimulations can be given in terms of relations and we did so. However, for two reasons we also give a coalgebraic definition that induces our relational definition. Firstly, due to the general framework our definition will cover a spectrum of bisimulations depending on the interpretation of the coalgebra and is applicable to more complex systems, automatically yielding the bisimulation definitions. Secondly, any ad-hoc features of a simple coalgebraic definition are more visible and can be clearly identified, whereas it is difficult to distinguish which of two similar relational definitions is more natural. As we assume no previous knowledge of categorical notions we give a brief introduction to coalgebras in the spirit of \cite{Sangiorgi:2011:ATB:2103601}.

A \emph{functor} $F$ (on sets) assigns to each set $X$ a set $F(X)$, and to each set function $f:X\to Y$ a set function $F(f):F(X)\to F(Y)$ such that two natural conditions are satisfied: (i) the identity function on $X$ is mapped to the identity function on $F(X)$ and (ii) a composition $f\circ g$ is mapped to a composition $F(f)\circ F(g)$.

\begin{example}
The powerset functor $\powerset(-)$ maps a set $X$ to the set $\powerset(X)$ of its subsets and a function $f:X\to Y$ to $\powerset(f):\powerset(X)\to\powerset(Y)$ by $U\mapsto \{f(x)\mid x\in U\}$. 

Similarly, for a fixed set $\act$, the operator $(-)^\act$ mapping $X$ to the set $X^\act$ of functions $\act\to X$ is a functor, where the image of $f:X\to Y$ is $F(f):X^\act\to Y^\act$ given by mapping $u:\act\to X$ to $f\circ u:\act\to Y$.
\end{example}

For a functor $F$, an \emph{$F$-coalgebra} is a pair of the carrier set (or state space) $S$ and the operation function $\nxt:S\to F(S)$. Intuitively, the function $\nxt$ describes the transition to the next step.

\begin{example}
A transition system $(S,\rightarrow)$ with $\mathord{\rightarrow}\subseteq S\times S$ can be understood as a $\mathcal P(-)$-coalgebra by setting $\nxt(s)=\{s'\mid s\tran{} s'\}$. And vice versa, every $\mathcal P$-coalgebra gives rise to a transition system.

A labelled transition system $(S,\act,\rightarrow)$ with the set of labels $\act$ and $\mathord{\rightarrow}\subseteq S\times \act\times S$ can be seen as a $(\powerset(-))^\act$-coalgebra with $\nxt:S\to (\powerset(S))^\act$ given by $\nxt(s)(a) = \{s'\mid s\tran{a}s'\}$.
\end{example}

A \emph{bisimulation} on an $F$-coalgebra $(S,\nxt)$ is a an $F$-coalgebra $(R,\overline\nxt)$ with $R\subseteq S\times S$ such that the two projections $\pi_1:R\to S$ and $\pi_2:R\to S$ make the following diagram commute:\footnote{I.e.\ $\nxt\circ\pi_1=F(\pi_1)\circ\overline\nxt$ and $\nxt\circ\pi_2=F(\pi_2)\circ\overline\nxt$.}

$$
\begin{tikzpicture}[x=2.5cm,y=1.2cm]
\node (S1) at (-1,0) {$S$};
\node (FS1) at (-1,-1) {$F(S)$};
\node (S2) at (1,0) {$S$};
\node (FS2) at (1,-1) {$F(S)$};
\node (R) at (0,0) {$R$};
\node (FR) at (0,-1) {$F(R)$};
\path[->] (S1) edge node[left] {$\nxt$} (FS1)
	  (S2) edge node[right] {$\nxt$} (FS2)
	  (R) edge node[left] {${\overline\nxt}$} (FR);
\path[->] (R) edge node[above] {$\pi_1$} (S1)
	  (R) edge node[above] {$\pi_2$} (S2)
	  (FR) edge node[below] {$F(\pi_1)$} (FS1)
	  (FR) edge node[below] {$F(\pi_2)$} (FS2);
\end{tikzpicture}$$

\begin{example}
For LTS, the coalgebraic bisimulation coincides with the classical one of Park and Milner \cite{DBLP:books/daglib/0067019}, where a symmetric $R$ is a bisimulation if for every $sRt$ and $s\tran{a}s'$ there is $t\tran{a}t'_{s,a,s',t}$ with $s'Rt'_{s,a,s',t'}$. Indeed, given a classical bisimulation $R$, one can define $\nxt(\langle s, t\rangle)(a)$ 
to contain for every $s\tran{a}s'$ the matching pair $\langle s',t'_{s,a,s',t}\rangle$ and symmetrically for $t$. Since all these pairs are from $R$, $(R,\overline\nxt)$ is indeed a coalgebra. Further, the projection $F(\pi_1)$ of $\nxt(\langle s, t\rangle)$ assigns to each $a\in \act$ all and nothing but the successors of $s$ under $a$, symmetrically for $t$, hence the commuting. 

Conversely, given a coalgebraic bisimulation $(R,\overline\nxt)$, the commuting of $\pi_1$ guarantees that $\nxt(\langle s, t\rangle)(a)$ yields all and nothing but the successors of $s$ under $a$. Hence, for each $s\tran{a}s'$ there must be $\langle s',t'\rangle\in\nxt(\langle s, t\rangle)(a)\subseteq R$, moreover, with $t\tran{a}t'$ due to $\pi_2$ commuting.
\end{example}

\medskip

As we have seen, the coalgebraic definition coincides with the relational one for non-probabilistic systems. One can use the same theory for finite probabilistic systems, too. Let $\dist(X)$ denote the set of simple distributions, i.e.\ functions $f:X\to[0,1]$ such that $f$ is non-zero for only finitely many elements $x_1,\ldots,x_n$ and $\sum_{i=1}^n f(x_i)=1$. Note that $\dist(-)$ can be understood as a functor. 

\begin{example}
We can interpret $(\dist(-)\cup\{\cannot\})^\act$-coalgebras as finite Markov decision processes $(S,\act,Pr)$ with $Pr:S\times \act\to\dist(S)\cup\{\cannot\}$ that under each action either proceed to a distribution on successors (as opposed to a non-deterministic choice in LTS) or not have the action available (the special element $\cannot$). The corresponding coalgebraic bisimulation can be shown to coincide with the classical one of Larsen and Skou~\cite{DBLP:conf/popl/LarsenS89}, where an equivalence relation $R$ is a bisimulation if $\sum_{u\in U}Pr(t,a)(u)=\sum_{u\in U}Pr(t',a)(u)$ for every $a\in \act$, classes $T,U$ of $R$ and $t,t'\in T$.
\end{example}

In contrast, uncountable probabilistic systems are more intricate. Let $\measures(X)$ now denote the set of \emph{probability measures} over $X$ (equipped with a $\sigma$-algebra clear from context).
Again, defining $\measures(f)(\mu)=\mu\circ f^{-1}$ makes $\measures(-)$ into a functor.

\begin{example}
We can interpret $\measures(-)$-coalgebras as Markov chains with general (possibly uncountable) state space. However, it is intricate to prove that the corresponding bisimulation is defined so that it coincides with the relational definition as already mentioned in Section~\ref{sec:intro}. 
\end{example}

\begin{example}
PA correspond to $(\powerset(\dist(-)))^\labels$-coalgebras. 
\end{example}

\subsection{Bisimulation on distributions coalgebraically}

The bisimulation we proposed is induced by a different view on the probabilistic systems. Namely, we consider 
distributions (or measures) $\measures(S)$ over its state space $\states$ to form the carrier of the coalgebra. A transition then changes this distribution. For instance, a Markov chain can be seen this way as a coalgebra of the identity functor.

Therefore, in order to capture the distributional semantics of NLMP and other continuous systems, we define a functor\footnote{On function, we define the functor by $\functorinf(f)(n)(A)=(id\times\powerset(f))(n(A))$. Here $\powerset(\labels)$ denotes only the measurable sets of labels.}
\begin{align*}
 \boxed{([0,1] \times \powerset(-))^{\powerset(\labels)} } \tag{$\functorinf$}
\end{align*}
The vital part is not only $[0,1]$, but also the use of measurable sets of labels instead of individual labels.
We can view a NLMP $\pts = (\states, \labels, \{\kernel_\lab \mid \lab\in\labels\})$ as a $\functorinf$-coalgebra with a carrier set $\measures(\states)$. The coalgebra assigns to $\mea \in \measures(\states)$ and to a set of labels $\Lab \in \labelsfield$ the pair $(p,\Mea)$ such that 
\begin{itemize}
 \item $p= \mea(\states_\Lab)$ is the measure of states that can read some $\lab\in\Lab$ where $\states_\Lab = \{ \sta\in\states \mid \exists \lab \in \Lab. \kernel_\lab(\sta) \neq \emptyset \}$;
 \item $\Mea = \emptyset$ if $\mea(\states_\Lab) = 0$, and $\Mea$ is the set of convex combinations\footnote{The set of convex combinations is lifted to a measurable set $Z$ of measures over $\states$ as the set $\{\Sta \mapsto \int_{\mea\in Z} \mea(\Sta) \meatwo(d\,\mea)  \mid \text{$\meatwo$ is a measure over $Z$}\}$.} over $\{\mea_\resolver \mid \text{measurable }\resolver:\states_\Lab\to\bigcup_{\lab\in\Lab}\kernel_\lab  \}$, otherwise, where
 $$ \mea_\resolver(\Sta) = \frac{1}{\mea(\states_\Lab)} \cdot \int_{\sta\in\states} \resolver(\sta)(\Sta)\ \mea(d\sta) \quad \forall \Sta\in\statesfield.$$
\end{itemize}
In other words, $M$ is obtained by restricting $\mu$ to the states that can read $\Lab$ and weighting all possible combinations of their transitions. 

\begin{lemma}\label{lem:inf-coin}
The 
union of $\functorinf$-bisimulations and $\bisim$ coincide.
\end{lemma}
\begin{proof}
First, we prove that whenever there is $\functorinf$-bisimulation $(R,\overline\nxt)$ with $(\dis,\distwo)\in R$ then $\dis\bisim\distwo$ by proving that $R\cup R^{-1}$ is a bisimulation relation. Let $A\subseteq\act$ and $\dis R \distwo$ or $\distwo R \dis$, w.l.o.g.\ the former (the latter follows symmetrically).
\begin{enumerate}
 \item The first condition of the relational bisimulation follows by
\begin{align*}
\dis(S_A)&=\pi_1(\nxt(\dis)(A))\\
&=\pi_1(\nxt\circ\pi_1\langle\dis,\distwo\rangle(A))\\
&=\pi_1(\functorinf\pi_1\circ\overline\nxt\langle\dis,\distwo\rangle(A))\\ &=\pi_1((id\times\powerset\pi_1)(\overline\nxt\langle\dis,\distwo\rangle(A)))\\
&=id(\pi_1(\overline\nxt\langle\dis,\distwo\rangle(A)))\\
&=\pi_1((id\times\powerset\pi_2)(\overline\nxt\langle\dis,\distwo\rangle(A)))\\
&=\pi_1(\functorinf\pi_2\circ\overline\nxt\langle\dis,\distwo\rangle(A))\\
&=\pi_1(\nxt\circ\pi_2\langle\dis,\distwo\rangle(A))\\
&=\pi_1(\nxt(\distwo)(A))\\
&=\distwo(S_A)
\end{align*}
 \item For the second condition of the relational bisimulation, let $\dis\tran{A}\dis'$. Since
\begin{align*}
\dis'&\in\pi_2(\nxt(\dis))(A)\\
&=\pi_2(\nxt\circ\pi_1\langle\dis,\distwo\rangle(A))\\
&=\pi_2(\functorinf\pi_1\circ\overline\nxt\langle\dis,\distwo\rangle(A))\\
&=\pi_2((id\times\powerset\pi_1)\Big(\overline\nxt\langle\dis,\distwo\rangle(A)\Big))\\
&=\powerset\pi_1(\pi_2\Big(\overline\nxt(\langle\dis,\distwo\rangle)(A)\Big))
\end{align*}
there is $\distwo'$ with
$$\langle\dis',\distwo'\rangle\in \pi_2\Big(\overline\nxt(\langle\dis,\distwo\rangle)(A)\Big)$$
Since $R$ is a coalgebra, we have $\langle\dis',\distwo'\rangle\in R$, i.e.\ $\dis'R\distwo'$.
\end{enumerate}

\medskip

Second, given $R=\mathord{\bisim}$, we define $\overline\nxt$ making it into a coalgebra such that the bisimulation diagram commutes. Let $\suc_A(\dis)=\{\dis'\mid\dis\tran{A}\dis'\}$ denote the set of all $A$-successors of $\dis$. For $\dis R\distwo$, we set
$$\overline\nxt(\langle\dis,\distwo\rangle)(A)=(\dis(S_A),\{\langle\dis',\distwo'\rangle\in R\cap \suc_A(\dis)\times\suc_A(\distwo)\})$$
Since we imposed $\langle\dis',\distwo'\rangle\in R$, $(R,\overline\nxt)$ is a $\functorinf$-coalgebra. Further, we prove the bisimulation diagram commutes. Firstly, 
\begin{align*}
\nxt\circ\pi_1\langle\dis,\distwo\rangle&=(\dis(S_A),\suc_A(\dis))\\ 
\nxt\circ\pi_2\langle\dis,\distwo\rangle&=(\distwo(S_A),\suc_A(\distwo))
\end{align*}
Therefore,
$$\pi_1(\nxt\circ\pi_1\langle\dis,\distwo\rangle)= \dis(S_A)=\pi_1(\functorinf\pi_1(\overline\nxt\langle\dis,\distwo\rangle)(A))$$
and
$$\pi_1(\nxt\circ\pi_2\langle\dis,\distwo\rangle)=\distwo(S_A)=\dis(S_A)= \pi_1(\functorinf\pi_2(\overline\nxt\langle\dis,\distwo\rangle)(A))$$
since $\dis(S_A)=\distwo(S_A)$ due to $\dis\bisim\distwo$ and the first relational bisimulation condition. Secondly,
\begin{align*}
\pi_2(\nxt\circ\pi_1\langle\dis,\distwo\rangle(A))&= \suc_A(\dis)\stackrel{(1)}=\pi_2(\functorinf\pi_1(\overline\nxt\langle\dis,\distwo\rangle)(A))\\
\pi_2(\nxt\circ\pi_2\langle\dis,\distwo\rangle(A))&=\suc_A(\distwo) \stackrel{(2)}=\pi_2(\functorinf\pi_2(\overline\nxt\langle\dis,\distwo\rangle)(A)) 
\end{align*}
After we show $(1)$ and $(2)$, we know both components of  $\functorinf\pi_1(\overline\nxt\langle\dis,\distwo\rangle)(A)$ are the same as of $\nxt(\pi_1\langle\dis,\distwo\rangle)(A)$, and similarly for $\functorinf\pi_2$, hence the commuting. As to $(1)$, $\supseteq$ follows directly by $\overline\nxt$ defined above. For $\subseteq$, for every $\dis'\in\suc_A(\dis)$ there is $\distwo'\in\suc_A(\distwo)$ with $\dis'R\distwo'$ due to the second realtional bisimulation condition. Thus also $\langle\dis',\distwo'\rangle\in \functorinf\pi_1 (\overline\nxt\langle\dis,\distwo\rangle)(A)$. $(2)$ follows from symmetric argument and $R$ being symmetric. \QED
\end{proof}

\subsection{Related bisimulations}

For \emph{discrete} systmes, one could define a functor for finite probabilistic systems with non-determinism by
\begin{align*}
{ ([0,1] \times \powerset(-))^\labels  \tag{$\functorfin$}}
\end{align*}

Now a PA $(\states, \labels, \patrans)$ is a $\functorfin$-coalgebra 
with the carrier set $\dist(\states)$. Indeed, the coalgebra assigns to a distribution $\dis$ and a label $\lab$ the pair $(p,\Dis)$ where 
\begin{itemize}
 \item $p= \dis(\states_\lab)$ is the probability of states that can read $\lab$;
\item $\Mea = \emptyset$ if $\mea(\states_\lab) = 0$, and $\Mea$ is the set of convex combinations over $\{\frac{1}{\dis(\states_\lab)}\sum_{\sta \in\states_\lab} \distwo_\sta \cdot \dis(\sta) \mid \forall \sta\in\states_\lab.\sta \patransa{\lab} \distwo_\sta\}$, otherwise.
We write $\dis \patransa{\lab} \dis'$ for every $\dis'\in \Dis$.
\end{itemize}

\begin{remark}
The union of $\functorfin$-bisimulations and bisimulation of \cite{DBLP:journals/corr/FengZ13}, denoted by $\lijunbisim$, coincide.
\end{remark}

Although we can use $\functorfin$ to capture the distribution semantics of PA as above, we could as well use it differently:
if we defined that a label that cannot be read in the current state is  \emph{ignored} instead of halting, the successor distribution would be defined by making a step from states that can read the label and staying elsewhere. (This approach is discussed in the next section.) 

\medskip

Moreover, we could easily extend the functor to systems with real rewards (as in~\cite{DBLP:conf/ijcai/CastroPP09}) 
simply by adding $\Rset$ to get $\Rset\times([0,1] \times  \powerset(-))^\labels $ for rewards on states or $([0,1] \times \powerset(\Rset\times -))^\labels $ on transitions etc. Similarly, for systems without the inner non-determinism like Rabin automata, we could simplify the functor to $([0,1] \times -)^\labels$. The only important and novel part of the functor is $[0,1]$ stating the overall probability mass that performs the step. (This is also the only difference to non-probabilistic coalgebraic functors.)
In all the cases, the generic $\functorfin$-bisimulation keeps the same shape. What changes is the induced relational bisimulation. 

\section{Applications}
\label{sec:applications-app}

In the following subsections, we justify the proposed 
bisimulation  yielded by 
$\functorinf$
by reviewing its application areas and comparing it to other bisimulations in these areas.

\subsection{Bisimulation in compositional modelling of distributed systems}

Probabilistic automata are apt for compositional modelling of communicating parallel systems. This way, the whole system is built bottom-up connecting smaller components into larger by the parallel composition operator. To tackle the state space explosion, minimisation algorithms can be applied throughout the process after each composition. 
Computing the quotient according to a bisimulation serves well as a minimisation algorithm if the bisimulation is a congruence w.r.t. parallel composition.
This condition is satisfied by the (also distribution-based) strong bisimulation recently defined by Hennessy~\cite{DBLP:dblp_journals/fac/Hennessy12}, denoted by $\bisim_\hennessy$. This is not the case with $\bisim$ as shown in the following example.

%
%
%
%
%

\begin{example}
According to our definition, $u \bisim v$ because $\frac{1}{2} u_h + \frac{1}{2} u_t \bisim v'$. In contrast, $u \not\bisim_\hennessy v$.
Therefore, $\bisim_\hennessy$ is strictly finer than $\bisim$. Actually, $\bisim_\hennessy$ coincides (on Dirac distributions) with the standard probabilistic bisimulation 
of Larsen and Skou \cite{DBLP:conf/popl/LarsenS89} which distinguishes $u$ and $v$ as well.

\begin{tikzpicture}[outer sep=0.1em,->,
state/.style={draw,circle,minimum size=1.6em,inner sep=0.1em}]

\begin{scope}[xscale=1.2]
\node (sm1) at (0.1,0) [state] {$v$};
\node (s0) at (1,0) [state] {$v'$};
\node (s0a) at (1.5,0) [inner sep=0, outer sep=0,minimum width=0] {};
\node (s1) at (2,0.6) [state] {};
\node (s2) at (2,-0.6) [state] {};

\draw [<-] (sm1) -- +(-0.5,0);
\draw (sm1) to node[auto,pos=0.6] {$a$} (s0);

\path[-] (s0) edge node[below=-1,pos=0.3]{$a$} (s0a);

\draw (s0a) to node[above left=-4,pos=0.7] {$\frac{1}{2}$} (s1);
\draw (s0a) to node[below left=-4,pos=0.7] {$\frac{1}{2}$} (s2);

\draw [loop right,looseness=5] (s1) to node[auto] {$h$} (s1);
\draw [loop right,looseness=5] (s2) to node[auto] {$t$} (s2);
\end{scope}

\begin{scope}[xscale=1.2,xshift=-10em]
\node (s0) at (0,0) [state] {$u$};
\node (s0a) at (0.5,0) [inner sep=0, outer sep=0,minimum width=0] {};
\node (s1) at (1,0.6) [state] {$u_h$};
\node (s2) at (1,-0.6) [state] {$u_t$};

\node (s3) at (1.9,0.6) [state] {};
\node (s4) at (1.9,-0.6) [state] {};

\draw [<-] (s0) -- +(-0.5,0);

\path[-] (s0) edge node[below=-1,pos=0.3]{$a$} (s0a);

\draw (s0a) to node[above left=-4,pos=0.7] {$\frac{1}{2}$} (s1);
\draw (s0a) to node[below left=-4,pos=0.7] {$\frac{1}{2}$} (s2);

\draw (s1) to node[auto,pos=0.6] {$a$} (s3);
\draw (s2) to node[auto,pos=0.6,swap] {$a$} (s4);

\draw [loop right,looseness=5] (s3) to node[auto] {$h$} (s3);
\draw [loop right,looseness=5] (s4) to node[auto] {$t$} (s4);
\end{scope}
\end{tikzpicture}


Let $\parallel_A$ denotes the $CSP$-style full synchronization on labels from $A$ and interleaving on $\labels \setminus A$. Then $\bisim$ is not a congruence w.r.t. $\parallel_A$ as $u \parallel_\labels s \not\bisim v \parallel_\labels s$ for $s$ depicted  below.
\vspace{-0.5em}
\begin{center}
\begin{tikzpicture}[outer sep=0.1em,->,
state/.style={draw,circle,minimum size=1.6em, inner sep=0.1}]

\begin{scope}[xscale=1.2]
\node (sm1) at (0.1,0) [state] {$s$};
\node (s0) at (1,0) [state] {$s'$};
\node (s1) at (2,0.6) [state] {};
\node (s2) at (2,-0.6) [state] {};

\draw [<-] (sm1) -- +(-0.5,0);
\draw (sm1) to node[auto,pos=0.6] {$a$} (s0);

\draw (s0) to node[above left=-4,pos=0.7] {$a$} (s1);
\draw (s0) to node[below left=-4,pos=0.7] {$a$} (s2);

\draw [loop right,looseness=5] (s1) to node[auto] {$h$} (s1);
\draw [loop right,looseness=5] (s2) to node[auto] {$t$} (s2);
\end{scope}
\end{tikzpicture}
\end{center}
\vspace{-0.3em} 

\end{example}

This is actually a classical example, due to~\cite{Segala:1996:MVR:239648}, modelling a process $u$ (or $v$) generating a secret by tossing a coin and the process $s$ guessing the secret. If $s$ guesses correctly, they synchronize forever on $h$ or $t$; otherwise, they halt. 
In $u \parallel_\labels s$, the non-determinism can be resolved by a \emph{scheduler} in such a way that the guesser makes a correct guess with probability $1$ which is not possible in $v \parallel_\labels s$ because the secret is generated later. This is overly pessimistic in the context of \emph{distributed systems} where the guesser observes only the communication with the tosser and not its state. Namely, the systems $u \parallel_\labels s$ and $v \parallel_\labels s$ exhibit the same behaviour (correct guess with probability at most $1/2$) if the non-determinism is resolved by \emph{distributed} schedulers~\cite{DBLP:conf/concur/AlfaroHJ01,Che06,DBLP:conf/formats/GiroD07}. This means that the non-determinism in each component of the composition is resolved independently of the state of the other component.

\subsection{Bisimulation for partially observable MDPs}

In the distributed setting it is natural to assume that the state space of each component is \emph{fully} unobservable from outside. This is a special case of \emph{partially} observable systems, such as partially observable Markov decision processes (POMDP). POMDPs have a wide range of applications in robotic control, automated planning, dialogue systems, medical diagnosis, and many other areas~\cite{DBLP:journals/aamas/ShaniPK13}.

In the analysis of POMDP, the distributions over states, called \emph{beliefs}, arise naturally and yield a continuous-space (fully observable) belief MDP. Therefore, probabilistic bisimulations over beliefs have been already studied~\cite{DBLP:conf/ijcai/CastroPP09,DBLP:conf/nfm/JansenNZ12}. However, no connection of this particular case to general probabilistic bisimulation has been studied. 

There are various (equivalent) definitions of POMDP, we use one close to computational game theory~\cite{DBLP:conf/mfcs/ChatterjeeDH10}.

\begin{definition}
 A \emph{partially observable Markov decision process (POMDP)} is a tuple $\mathcal{M} = (\states,\delta, \mathcal{O})$ where $\states$ is a set of states, $\delta \subseteq \states \times \dist(\states)$ is a transition relation, and $\mathcal{O} \subseteq 2^\states$ is a set of observations that partition the state space. 
\end{definition}

This formalism is also known as \emph{labelled Markov decision processes}~\cite{DBLP:journals/ijfcs/DoyenHR08} where state labels correspond to observations. 
Such a state-labelled system $\mathcal{M} = (\states,\delta, \mathcal{O})$ can be easily translated to an action-labelled PA $\PA_{\mathcal{M}} = (\states,\mathcal{O}, \patrans)$ where $\sta \patransa{o} \dis$ if $\sta \in o$ and $(\sta,\dis) \in \delta$. This way, we can define $\dis \bisim \dis'$ in $\mathcal{M}$ if $\dis \bisim \dis'$ in $\PA_\mathcal{M}$.

Hence, in Section~\ref{sec:algorithms-finite}, we give the first algorithm for computing bisimulations over beliefs in finite POMDP. Previously, there was only an algorithm~\cite{DBLP:conf/nfm/JansenNZ12} for computing bisimulations on distributions of Markov chains with partial observation. 


\subsection{Bisimulation for large-population models}

In the sense of~\cite{DBLP:journals/deds/GastG11,DBLP:journals/tcs/McCaigNS11,may1974biological,jovanovic1988anonymous}, we can understand PA as a description of one \emph{agent} in a large homogeneous population. For example a \emph{chemical compounds}, a \emph{node of a computer grid}, or a \emph{customer of a chain store}.
%

The distribution perspective is a natural one -- the distribution specifies the ratios of agents being currently in the individual states. For a Markov chain, this gives a deterministic process over the continuous space of distributions. 

The non-determinism of PA has also a natural interpretation.
Labels given to this large population of PAs correspond to global control actions~\cite{DBLP:journals/tac/GastGB12,DBLP:journals/deds/GastG11} such as \emph{manipulation with the chemical solution}, a \emph{broadcast within the grid}, or a \emph{marketing campaign of the chain store}. Agents react to this control action if currently in a state with transition under this label, otherwise they ignore it.
Multiple transitions under this label correspond to multiple ways how the agent may react. 


\begin{example}
Let us illustrate the idea by an example of three models of customers of a chain store with half of the population in state $1$ and half of the population in state $3$. 
\begin{center}
\begin{tikzpicture}[outer sep=0.1em,->,yscale=0.9,
state/.style={draw,circle, minimum size=1.7em,inner sep=0.1em
,inner sep =0.1em,text centered},
trans/.style={font=\scriptsize\itshape},
]

\begin{scope}[xscale=1,yscale=0.8]
%
%
\node (sm0) at (1.75,0) [state] {2};
\node (sm1) at (0,0) [state] {1};

\node (s1) at (2.5,-2) [state] {4};

\node (s2) at (1,-2) [state] {3};



\draw (sm0) to node[trans,auto,pos=0.7] {yoghurt ad} (s1);
\draw (s1) [loop below,looseness=5] to node[trans,auto] {buy y.} (s1);

\draw (sm0) to node[trans,auto,swap,pos=0.7] {m\"{u}ssli ad} (s2);
\draw (s2) [loop below, looseness=5] to node[trans,auto] {buy m.} (s2);

\draw [-] (3.8,0.5) -- (3.8,-3.2);

\end{scope}

\begin{scope}[xscale=1,yscale=0.8,xshift=15em]
%
%
\node (sm0) at (1.75,0) [state] {2};
\node (sm1) at (0,0) [state] {1};

\node (s1) at (1.75,-2) [state] {4};

\node (s2) at (0,-2) [state] {3};



\draw (sm0) to node[trans,auto,swap] {yoghurt ad} (s1);
\draw (s1) [loop below,looseness=5] to node[trans,auto] {buy y.} (s1);

\draw (sm1) to node[trans,auto,swap] {m\"{u}ssli ad} (s2);
\draw (s2) [loop below, looseness=5] to node[trans,auto] {buy m.} (s2);


\draw [-] (2.3,0.5) -- (2.3,-3.2);
\end{scope}

\begin{scope}[xscale=1,yscale=0.8,xshift=23.5em]
%
%
\node (sm0) at (2,0) [state] {2};
\node (sm1) at (0,0) [state] {1};

\node (s1) at (2.75,-1.5) [state] {4};

\node (s2) at (1.25,-1.5) [state] {3};
\node (s3) at (2,-3) [state] {5};



\draw (sm0) to node[trans,right,pos=0.6] {yoghurt ad} (s1);
\draw (s1) [loop right,looseness=3] to node[trans] {buy y.} (s1);

\draw (sm0) to node[trans,left,swap,pos=0.6] {m\"{u}ssli ad} (s2);
\draw (s2) [loop left, looseness=3] to node[trans,left=-2] {buy m.} (s2);

\draw (s1) to node[trans,auto,pos=0.1] {m\"{u}ssli ad} (s3);
\draw (s2) to node[trans,auto,pos=0.1,swap] {yoghurt ad} (s3);

\draw (s3) [loop right,looseness=3] to node[trans] {buy y.} (s3);
\draw (s3) [loop left, looseness=3] to node[trans,left=-2] {buy m.} (s3);

\end{scope}

\end{tikzpicture}
\end{center}
It is natural to assume that these three models can be distinguished. Indeed none of the populations are bisimilar according to our definition.
Note however, that the related distribution-based bisimulation of~\cite{DBLP:journals/corr/FengZ13} that allows only singletons $\Lab$ in Definition~\ref{def:infinite-bisim} does not distinguish the first and the second population. Their definition actually extends the bisimulation of~\cite{DBLP:journals/ijfcs/DoyenHR08} defined on input-enabled models; they naturally transform general probabilistic automata to input-enabled ones by directing the missing transitions into a newly added sink state. Observe that the similarly natural alternative approach of adding self-loops does not distinguish the second and the third population. 
\end{example}

\section{Technical details and proofs from Section~\ref{sec:applications}}

Let us first formalize in more detail the concepts we relate to in the main body.

\subsection{Continuous-time Markov chains}

\begin{definition}
A CTMC $\mathcal{C}$ is a tuple $(S,Q)$ where $S$ is a finite set of states, and $Q: S\times S \to \Rsetpo$ is a rate matrix such that $Q(s,s'^) = 0$ denotes that there is no transition from $s$ to $s'$.
\end{definition}

\subsubsection{Parallel composition} For two CTMC $(S_1,Q_1)$ and $(S_2,Q_2)$ with initial states $s_1$ and $s_2$ we define their (full interleaving) parallel composition $\mathcal{C}_1 \parallel_{CT} \mathcal{C}_2$ as $(S_1\times S_2, Q')$ with the initial state $(s_1,s_2)$  where 
$$Q((s_1,s_2),(s'_1,s'_2)) = \begin{cases}
 Q_1(s_1,s'_1) & \text{if $s_2 = s'_2$,} \\
 Q_1(s_1,s'_1) & \text{if $s_2 = s'_2$,}  \\
 0 & \text{otherwise.}
 \end{cases}
 $$
 
\subsubsection{Embedding} Finally, to each CTMC $\mathcal{C} = (S,Q)$ with initial state $s_0 \in S$, 
we define a stochastic automaton $SA(\mathcal{C}) = (S,\states\times\states,\{\act\},\edges,\clocksetting, \clockdist)$ with initial location $s_0$ where 
  \begin{itemize}
  \item $(s,\act,\{(s,s')\},s') \,\in\; \edges$ for any $s,s' \in \states$,
  \item $\clocksetting(s) = \{(s,s') \mid Q(s,s') > 0 \}$,
  \item $\clockdist((s_1,s_2)) = Exp(Q(s_1,s_2))$ 
  \end{itemize}

\subsection{Stochastic automata}

\subsubsection{Semantics $\pts_\SA$ of stochastic automata}
Let $\SA = (\locations,\clocks,\actions,\edges,\clocksetting, \clockdist)$ be a stochastic automaton with initial location $\loc_0$. 
We define the semantical NLMP $\pts_\SA = (\locations \times
\Rset^\clocks, \actions \times \Rsetpo, \{\kernel_\lab \mid \lab\in\labels\})$.
A state $(\loc,\val)$ denotes being in location $\loc$ where each clock $c$ has value $\val(c)$.
The NLMP $\pts_\SA$ is initiated according to a initial measure $\mea$ over the state space of $\pts_\SA$ such that
\begin{itemize}
\item the marginal in the first component being Dirac on $\loc_0$;
\item the marginal for any $\clo \not\in \clocksetting(\loc_0)$ being Dirac on $0$;
 \item the marginals for each $\clo \in \clocksetting(\loc_0)$ having CDF $\clockdist(\clo)$, and their product being equal to the joint distribution of $\clocksetting(\loc_0)$.
\end{itemize}
In $(\loc,\val)$, a label of the form $(a,t)$ is available if $E_a \neq \emptyset$ where $E_a$ is the set of edges that have action $a$ and become available after the idling time $t$.
We set $\kernel_{(a,t)}((\loc,\val)) = \{\mu_e \mid e \in E_a \}$ where $\mu_e$ for an edge $e = (q,a,C,q')$ is the probability measure over states with (similarly to the previous case)
\begin{enumerate}
 \item the marginal in the first component being Dirac on $\loc'$;
 \item the marginal for any $\clo \not\in \clocksetting(\loc')$ being Dirac on $\val(\clo) - t$;
 \item the marginals for each $\clo \in \clocksetting(\loc')$ having CDF $\clockdist(\clo)$, and their product being equal to the joint distribution of $\clocksetting(\loc')$.
\end{enumerate}
Intuitively, it (1) moves to $\loc'$, (2) decreases values of clocks by $t$, and (3) sets clocks of $\clocksetting(\loc')$ to independent random values.

 \subsubsection{Parallel composition} Further, for two SA $\SA_1 = (\locations_1,\clocks_1,\actions_1,\edges_1,\clocksetting_1, \clockdist_1)$ and $\SA_2 = (\locations_2,\clocks_2,\actions_2,\edges_2,\clocksetting_2, \clockdist_2)$ with initial locations $\loc_1$ and $\loc_2$ we define their full interleaving parallel composition $\SA_1 \parallel_{S\!A} \SA_2$ as the tuple $(\locations_1 \times \locations_2 \times \{0,1,2\},\clocks_1 \cup \clocks_2,\actions_1 \cup \actions_2,\edges,\clocksetting, \clockdist)$ with initial location $(\loc_1,\loc_2,0)$, where the third component of a location denotes which of the two SA moved the last step and where 
 \begin{itemize}
 \item $\edges$ is the smallest relation satisfying 
 \begin{itemize}
 \item $(\loc, \act, \Clo,\loc') \in \edges_1$ implies $((\loc,\loc_2,b),\act, \Clo, (\loc',\loc_2,1\{)) \in \edges$ for any $\loc_2 \in \locations_2$ and $b\in \{0,1,2\}$ and 
 \item $(\loc, \act, \Clo,\loc') \in \edges_2$ implies $((\loc_1,\loc,b),\act, \Clo, (\loc_1,\loc',2)) \in \edges$ for any $\loc_1 \in \locations_1$ and $b\in\{0,1,2\}$;
 \end{itemize}
 \item $\clocksetting((\loc_1,\loc_2,b)) = \clocksetting_b(\loc_b)$ if $b \in \{1,2\}$ and $\clocksetting((\loc_1,\loc_2,0)) = \clocksetting_1(\loc_1) \cup \clocksetting_2(\loc_2)$,
 \item $\clockdist$ assigns $\clockdist_1(\clo)$ to $\clo \in \clocks_1$ and $\clockdist_2(\clo)$ from $\clo \in \clocks_2$.
 \end{itemize}
 
\subsection{Proof of Theorem~\ref{thm:sta-commute}} Let us recall the theorem.

\begin{reftheorem}{thm:sta-commute}
  Let $SA(\mathcal{C})$ denote the stochastic automaton corresponding
  to a CTMC $\mathcal{C}$.  For any CTMC $\mathcal{C}_1,
  \mathcal{C}_2$, we have
 $$SA(\mathcal{C}_1) \parallel_{S\!A} SA(\mathcal{C}_1) 
 \;\; \bisim \;\;  
 SA(\mathcal{C}_1 \parallel_{CT} \mathcal{C}_1).$$ 
\end{reftheorem}
\begin{proof}
It is easy to see that to each location $(s_1,s_2)$ in the system on the right there are three locations of the form $(s_1,s_2,b)$ in the system on the left, that differ only in the third component $b$, i.e. they
\begin{itemize}
\item have the same set of edges,
\item have the same set $Pos(s_1,s_2)$ of clocks that are positive in each location,
\end{itemize}
and differ only in the sets of clocks $\clocksetting$ to be re-sampled. 

We show that $(s_1,s_2) \bisim (s_1,s_2,b)$ for any $b \in \{0,1,2\}$ by applying the arguments from the algorithm in Section~\ref{sec:algorithms-infinite}.
Let $\chain_L$ and $\chain_R$ denote the finite systems from Lemma~\ref{lem:expo-finite} obtained from the systems on the left and on the right, respectively.
The distribution of clocks in each location $\loc = (s_1,s_2)$ or $\loc = (s_1,s_2,b)$ is $\loc \otimes \bigotimes_{c\in Pos(s_1,s_2)} Exp(\lambda_c)$.
Hence, to each state on the right, there are at most 3 reachable states on the left with the same clock distributions. 
Thanks to the same edges and same clock distributions, these three states are indistinguishable by the \tableau{Step} rule.\QED

\end{proof}

\section{Proofs from Section~\ref{sec:algorithms}}

\subsection{Discrete systems}

We use the notation $\mu\oplus_p\nu$ to denote $(1-p)\mu+p\nu$. Further, for a (not necessarily probabilistic) measure $\dis=(\dis(s_1),\ldots,\dis(s_{|\states|}))$ we denote $|\dis|=\sum_{i=1}^{|\states|}\dis(s_i)$. For any probability distribution $\dis$ thus $|\dis|=1$.

\begin{reflemma}{lem:existence}
For every linear bisimulation there exists a corresponding bisimulation matrix.
\end{reflemma}
\begin{proof}
Let $R$ be a linear bisimulation and
$\Gamma$ an arbitrary equivalence class of $R$.
Due to linearity, $\Gamma$ is closed under convex combinations. Consider $\bar\Gamma$ the affine closure of $\Gamma$, i.e. the smallest set that is closed under affine combinations. Then (i) $\bar\Gamma$ is an affine subspace, and (ii) $\bar\Gamma\cap\measures(\states)=\Gamma$. This holds for every class of $R$. Hence $\{\bar\Gamma\mid\Gamma \text{ is an equivalence class of } R\}$ decomposes  $\mathbb R^{|\states|}$ and all $\bar\Gamma$ have the same difference space $\bar\Delta:=\{\mu-\nu\mid \mu,\nu\in\bar\Gamma\}$ (independent of choice of $\Gamma$). Since $\bar\Delta$ is a linear subspace, there is a matrix $\B$ such that $\rho\in\bar\Delta$ iff $\rho\B=0$.

For every $\dis R \distwo$ we thus have $(\dis-\distwo)\B=0$. 
In the other direction, let $\dis\in\Gamma$ and $\distwo$ be arbitrary distribution such that $(\dis-\distwo)\B=0$. We thus have $\dis-\distwo\in\bar\Delta$. Since $\dis\in\bar\Gamma$ we thus get $\distwo\in\bar\Gamma$. Since $\distwo\in\measures(\states)$, we finally obtain $\distwo\in \Gamma$ and thus $\dis R \distwo$. \QED
\end{proof}

\begin{lemma}\label{lem:linearity}
$\bisim$ is linear.
\end{lemma}
\begin{proof}
We prove that $\dis_1\sim\distwo_1$ and $\dis_2\sim\distwo_2$ imply $\dis_1\oplus_p\dis_2\sim\distwo_1\oplus_p\distwo_2$ for any $p\in[0,1]$. This follows easily from the Spoiler-Duplicator game. Indeed, let Duplicator have a winning response to every Spoiler's strategy both in $\dis_1\sim\distwo_1$ and $\dis_2\sim\distwo_2$. Let now $p\in[0,1]$. Any Spoiler's strategy on $\dis_1\oplus_p\dis_2\sim\distwo_1\oplus_p\distwo_2$ (w.l.o.g.\ attacking on the left under $A$) can be decomposed to a part acting on $(1-p)\dis_1$ resulting into $\Big((1-p)\dis_1(S_A),(1-p)\dis_1'\Big)$ and a part acting on $p\distwo$ resulting into $\Big(p\dis_2(S_A),p\dis_2'\Big)$. Duplicator has a winning response $\distwo_1'$ to the former (when applied to the whole $\dis_1$) and also $\distwo_2'$ to the latter (when applied to the whole $\dis_2$). Duplicator can now mix his 
responses resulting into $\distwo_1'\oplus_p\distwo_2'$, which is clearly a choice conforming both to the rules, since $(\distwo_1\oplus_p\distwo_2)(S_A)=(1-p)\distwo_1(S_A)+p\distwo_2(S_A)=(1-p)\dis_1(S_A)+p\dis_2(S_A)=(\dis_1\oplus_p\dis_2)(S_A)$ and also winning as the resultinig pair is again a convex combination of individual resulting pairs.\QED
\end{proof}

Thus minimal bisimulation matrices always exist.

\begin{corollary}
There is a minimal bisimulation matrix, i.e.\ a matrix $\B$ such that for any $\dis,\distwo\in\measures(\states)$, we have  $\dis\sim\distwo$ iff $(\dis-\distwo)\B=0$.
\end{corollary}

We are searching for the least restrictive system $\B$ satisfying stability. Therefore, we can compute $\bisim$, i.e.\ the greatest fixpoint of the bisimulation requirement of stability, as the least fixpont of the partitioning procedure of adding equations. Indeed, recall that all bisimulation matrices with the least possible dimension have the same solution space.

\begin{refproposition}{prop-alg-determ}
In an action deterministic PA, $\B$ containing $\vec 1$ is a bisimulation matrix iff it is  $P_\lab$-stable for all $\lab\in\labels$.
\end{refproposition}
\begin{proof}
Firstly, we prove that for any $\lab\in\labels$, any bisimulation matrix $\B$ is $P_\lab$-stable. Let $\rho$ be such that $\rho\B=0$. Let us write $\rho=\dis-\distwo$ where entries in $\dis$ and $\distwo$ are non-negative. Since $\B$ contains $\vec1$, we have $|\dis|=|\distwo|$, moreover, for the moment assumed, equal 1. Then $\rho$ is a difference of two measures $\dis-\distwo$. Since $\B$ is a bisimulation matrix, we have $\dis\bisim\distwo$. Therefore, if Spoiler attacks under $a$, we have $\dis P_a\bisim\distwo P_a$. Therefore, $(\dis P_a-\distwo P_a)E=0$, equivalently $\rho P_aE=0$. In the general case, where $|\dis|=|\distwo|$ is not equal $1$, we can egard them as a scalar multiples of measures, normalize them, and use the same reasoning (with the exception when they are $\vec0$, in which case the claim for $\rho=\vec0$ holds trivially).

Secondly, let $\B$ contain $\vec1$ and be $P_\lab$-stable for all $\lab\in\labels$. We show that $R$ defined by $\dis R\distwo$ iff $(\dis-\distwo)\B=0$ is a bisimulation relation. Consider now $A\subseteq\act$ singletons. The first bisimulation condition for $a\in\act$ follows from $(\dis-\distwo)P_a\vec1=0$. The second one then from $(\dis-\distwo)P_a\B=0$ implying $(\dis P_a-\distwo P_a)\B=0$ by stability. For general $A\subseteq\act$, the bisimulation condition does not generate any new requirements due to the action determinism. Since $S_A$ is a disjoint union of $S_a$ for $a\in A$, the properties follow from the properties of singeltons. \QED
\end{proof}


We recall that for elements of $\rohy$ are tuples of corners of $C_i$'s that are ``extremal in the same direction.''
Formally, we say a point $p$ is \emph{extremal in direction $d$} (in a polytope $P$) if $d$ is a normal vector of a separating hyperplane containing only $p$ from the whole $P$ and such that $p+d$ lies in the other half-space than $P$. 

Intuitively, elements of $\rohy$ are those tuples of corners that form corners of ``combinations'' of $C_i$'s. 
Formally, denote the $|\states|$-dimensional vector of $C_i$'s by $\vec C$.
For a distribution $\mu$, the ``$\mu$-combination of polytopes $C_i$'' is the polytope 
$$\mu\vec C^\top = \{\sum_{i=1}^{|\states|} \mu(s_i){c_i} \mid \forall i: {c_i}\in C_i\}$$ The corners $\corners(\mu\vec C^\top)$ are then exactly $\{\mu c^\top\mid c\in \rohy\}$.

Further, we call that a choice is \emph{extremal} if it can be written as $W(c)$ for some extremal $c$, i.e. $c\in \rohy$.
Note that these points are mapped to pure strategies and achieve \emph{Pareto extremal values} when applied to any distributions, i.e. $\mu c^\top$ is a corner of $\mu \vec C^\top$ for every distribution $\mu$.

\begin{refproposition}{prop-alg-nondeterm}
$\B$ containing $\vec 1$ is a bisimulation matrix iff the matrix is $\PaWc$-stable for all $A\subseteq\labels$ and $c\in\rohy$.
\end{refproposition}
\begin{proof}
 Observe that if $\dis\bisim\distwo$ then $\dis\vec C^\top$ and $\distwo \vec C^\top$ are the same polytopes. Indeed, for every choice on one side there must be a choice on the other side matching in all components. Conversely, if $\dis\vec C^\top\neq\distwo \vec C^\top$ then $\dis\not\bisim\distwo$ as Spoiler can choose a vector that cannot be matched by Duplicator. Note that equality of polytopes $\dis\vec C^\top$ and $\distwo \vec C^\top$ can be tested by equality of the sets of their extremal points. The extremal points are exactly points $\dis \vec c^\top$ and  $\distwo \vec c^\top$ for $\vec c\in \rohy$. 

Hence we prove the two following facts:
\begin{enumerate}
 \item[(1)] the extremal choices, i.e.\ $\rohy$, are sufficient for Spoiler, 
 \item[(2)] for an extremal choice $W\in\rohy$ of Spoiler, $W$ is an optimal reply of Duplicator for any distributions $\dis$ and $\distwo$. 
\end{enumerate}
As to (1), intuitively, if two polytopes are different, there must be a corner of one not in the other by convexity of the polytopes. Formally, for given $\dis\not\bisim\distwo$, $\dis\vec C^\top\neq\distwo \vec C^\top$ and an optimal choice of Spoiler is a $W(c)$ such that $\dis c^\top\notin \distwo \vec C^\top$ (or the other way round, $\distwo c^\top\notin \dis \vec C^\top$).
Such a choice can be done so that $\dis M_A^{W(c)}$ is Pareto extremal hence corner of $\dis\vec C^\top$. 

As to (2), intuitively, if two polytopes are the same and Spoiler checks whether a corner $c_1$ of one is also a corner of the other, Duplicator has to answer with a corner $c_2$ that is extremal in the same direction as $c_1$. Formally, let $\dis\bisim\distwo$ and $W(s)$ be an extremal choice of Spoiler on $\dis$, $W(d)$ an optimal (winning) response of Duplicator on $\distwo$ supposed, for a contradiction, different from $W(s)$. Since $s$ is extreme in some direction $v$ for which $d$ is not, and since $W(s)$ achieves on $\dis$ the same as $W(d)$ on $\distwo$, there is a choice $W(d')$ where $d'$ is extremal in direction $v$ and thus achieves strictly better Pareto value on $\distwo$ than $d$, hence also strictly better (in direction $v$) than $W(s)$ on $\dis$. 

Now if Spoiler moved from $\distwo$ by $W(d)$ a matching response would be $W(s)$. On the other hand, if Spoiler moved from $\distwo$ by $W(d')$, this choice strictly dominates $W(d)$ on $\distwo$ (in direction $v$) and thus all choices on $\dis$ (in direction $v$) as $s$ is extremal in direction $v$. Hence there is no matching response for the Duplicator, a contradiction.

\medskip

As a result of (1) and (2), the bisimulation matrix requirement can be simplified. In the game fashion it is written as follows: for all $A\subseteq\act$
\begin{multline*}
(\dis - \distwo)\B=0 \implies\forall W_S\in\choices: \exists W_D\in\choices:\\\dis P_A^{W_S}\vec1=\distwo P_A^{W_D}\vec1 \wedge 
(\dis P_A^{W_S} - \distwo P_A^{W_D})\B=0 
\end{multline*}
Now we can transform it into: for all $A\subseteq\act$
\begin{multline*}
(\dis - \distwo)\B=0 \implies\forall W\in\rohy:\\ (\dis P_A^{W}-\distwo P_A^{W})\vec1=0 \wedge (\dis P_A^{W} - \distwo P_A^{W})\B=0
\end{multline*}
and since $\vec1$ is a column of $\B$, we can also write it equivalently  as: for all $A\subseteq\act$
$$(\dis - \distwo)\B=0 \implies 
\forall W\in\rohy: (\dis - \distwo) P_A^{W} \B=0$$
which is nothing but $P_A^{W(c)}$-stability for all $A\subseteq\labels$ and $c\in\rohy$. (We deal with $\rho$ not being a difference of any two distributions by scaling as in Proposition~\ref{prop-alg-determ}).
\QED
\end{proof}

\begin{corollary}
Any matrix $P_A^{W(c)}$-stable for all $A\subseteq\labels$ and $c\in\rohy$ and containing $\vec 1$ with minimal rank is a minimal bisimulation matrix.
\end{corollary}

\begin{reftheorem}{thm:algorithm-finite}
Algorithm~\ref{alg-fin} computes a minimal bisimulation matrix in exponential time.
\end{reftheorem}
\begin{proof}
The proof follows from the previous corollary and the fact that the algorithm only adds columns required by stability on the current partitioning.

Concerning the complexity, each step is polynomial except for computing and iterating over all exponentially many extremal choices and exponentially many sets of labels. 

The extremal points $\rohy$ can be computed easily: firstly, we identify which directions the corners of each $C_i$ are extremal for. The elements of $\rohy$ are combinations of corners etremal in the same direction. Therefore, we only need to compute the common partitioning of the directions according to extremality w.r.t. each corner.\QED
\end{proof}

%

\subsection{Continuous-time systems}
\label{app:infinite}

Let us repeat the main theorem of the subsection.

\begin{reftheorem}{thm:tableau}
Let $\SA = (\locations,\clocks,\actions,\edges,\clocksetting, \clockdist)$ be a deterministic SA over exponential distributions.
There is an algorithm to decide
in time polynomial in $|\SA|$ and exponential in $|\clocks|$ whether $\loc_1 \bisim \loc_2$ 
for any locations $\loc_1,\loc_2$.
\end{reftheorem}

\noindent
The proof follows easily from the following lemmata.

\begin{reflemma}{lem:abs}
For any distributions $\dis, \distwo$ on $\SAstates$ we have $\dis \bisim \distwo$ iff $\xi(\dis) \bisim \xi(\distwo)$.
\end{reflemma}

\begin{proof}
$\Rightarrow$: Let us take the maximal bisimulation in $\pts_\SA$. We map it by $\xi$; it is easy to see that it is still a bisimulation since the operations $\xi$ and $\patransa{A}$ commute for any $A\subseteq \labels$: for any distribution $\dis$, we have $\dis(\states_A) = \xi(\dis)(\states_A)$, and the unique distributions $\dis',\dis''$ such that $\dis \patransa{A} \dis'$ and $\xi(\dis) \patransa{A} \dis''$ satisfy $\dis'' = \xi(\dis')$.

$\Leftarrow$: Let us take $\dis$, $\distwo$ such that $\dis \not\bisim\distwo$. Then there is a finite sequence of set of labels $\Lab_1, \ldots, \Lab_n$, such that after applying this sequence, one of the conditions in Definition~\ref{def:infinite-bisim} is not satisfied. Again, as the operations $\xi$ and $\patransa{A}$ commute for any $A\subseteq \labels$, we get that also $\xi(\mu) \not\bisim \xi(\nu)$. \QED
\end{proof}

\begin{reflemma}{lem:expo-finite}
For a deterministic SA over exponential distributions, $|\chainstates| \leq |\locations|2^{|\clocks|}$.
\end{reflemma}

\begin{proof}
It is easy to check that for all states of the form $\loc \otimes \bigotimes_{\clo\in\Clo \subseteq \clocks} Exp(\lambda_\clo)$,
any successor in $\chain$ has the same form. Let us fix a state of such a form $q \otimes \bigotimes_{\clo \in \Clo} Exp(\lambda_c)$ and 
\begin{itemize}
\item an edge $(\loc,\act, \Clo', \loc')$ such that $\Clo \cap \Clo' = \{\clo'\}$ (i.e. exactly one clock from the trigger set is still positive). The successor state is of the form $\loc' \otimes \bigotimes_{\clo \in (\Clo\setminus \{\clo'\}) \cup \clocksetting(\loc')} Exp(\lambda_\clo)$. Indeed, the distribution $P[X - Y > t \mid X > Y]$ for $X \bisim Exp(\lambda)$ and $Y \bisim Exp(\mu)$ is still exponentially distributed with rate $\lambda$.
\item an edge of the general form $(\loc,\act, \Clo', \loc')$ such that $\Clo \cap \Clo' \neq \emptyset$ (i.e. some clocks from the trigger set are still positive) can be split into a diamond of edges among intermediate states when each clock from the set $\Clo \cap \Clo'$ runs down to zero, each of the intermediate states are of the specified form.\QED
\end{itemize} 
\end{proof}

\begin{reflemma}{prop:correctness}
There is a successful tableau from $\dis \bisim \distwo$ iff $\dis \bisim \distwo$ in $\abssemantics$.
Moreover, the set of nodes of a successful tableau is a subset of a bisimulation.  
\end{reflemma}
\begin{proof}
$\Leftarrow$: We can build an infinite successful tableau only using the rule \tableau{Step}. Note that the rule exactly follows the transition relation of $\abssemantics$ (only regards the distribution as a discrete convex combination of one of finitely many distributions -- states of $\chain$). Hence, by applying the rule \tableau{Step} from bisimilar distributions, we can obtain only tableau nodes corresponding to bisimilar distributions never reaching a failure node.

$\Rightarrow$: First, observe that if there is a successful tableau $T$ from node $\mu \bisim \nu$, there also is a successful (possibly infinite) tableau $T'$ using only the rule \tableau{Step}. This is easy to observe since whenever there is an application of the \tableau{Lin} rule, one can iteratively apply the \tableau{Step} rule infinitely many times (since one can express the current node as a linear combination of nodes from which one can apply the \tableau{Step} rule; and the same inductively holds for each such successor node). 

Note that by this construction, the set of nodes of $T$ is a subset of the set of nodes of $T'$. We show that for any node $\mu_1 \bisim \mu_2$ in $T'$ we have $\mu_1 \bisim \mu_2$ in $\abssemantics$.\
Let us fix such a node $\mu_1 \bisim \mu_2$ and let $R$ be a relation such that $\mu_1' R \mu_2'$ if $\mu_1' \bisim \mu_2'$ is an ancestor of the node $\mu_1 \bisim \mu_2$. Since the rule \tableau{Step} closely follows the definition of bisimulation, it is easy to see that $R$ is a bisimulation. As $R$ contains also $(\mu_1,\mu_2)$, we have $\mu_1 \bisim \mu_2$.
\QED
\end{proof}

\begin{reflemma}{lem:finite-tableau}
There is a successful tableau from $\dis \bisim \distwo$ iff there is a finite successful tableau from $\dis \bisim \distwo$ of size polynomial in $|\chainstates|$.
\end{reflemma}
\begin{proof}
The implication $\Leftarrow$ is trivial. As regards $\Rightarrow$, let us assume that there is a successful tableau from $\dis \bisim \distwo$. As each node in the tableau corresponds to a vector of dimension $|\chainstates|$, the maximal size of a set of linearly independent nodes is $|\chainstates|$. By applying the rule \tableau{Lin} when possible we can prune the tableau into linear size.
\QED
\end{proof}

Note that we not only have a polynomial bound on the size of a successful tableau, we also have a deterministic polynomial time  procedure to construct such a tableau. We build the tableau in arbitrary fixed order (such as breath-first) For each node, we first check whether the \tableau{Lin} rule can be applied; if not, we apply the \tableau{Step} rule. This concludes the proof of Theorem~\ref{thm:tableau}.